\documentclass{article}

\usepackage{a4wide}
\usepackage{times}
\usepackage{latexsym}
\usepackage{amssymb,amsmath,amsthm}
\usepackage{url}
\usepackage{enumerate}

\usepackage{tikz}
\usetikzlibrary{positioning,arrows,fit,calc}
\usetikzlibrary{decorations.pathmorphing}

\usepackage{pifont}

\newtheorem{theorem}{Theorem}
\newtheorem{proposition}[theorem]{Proposition}
\newtheorem{corollary}[theorem]{Corollary}
\newtheorem{lemma}[theorem]{Lemma}

\newtheorem{example}[theorem]{Example}

\newcommand{\OWL}{\textsl{OWL\,2}}
\newcommand{\OWLQL}{\textsl{OWL\,2\,QL}}

\newcommand{\q}{{\boldsymbol q}}
\newcommand{\pq}{\boldsymbol{p}}

\newcommand{\NLogSpace}{\textsc{NL}}
\newcommand{\coNLogSpace}{\textsc{coNL}}
\newcommand{\LogSpace}{\textsc{L}}

\newcommand{\NP}{\textsc{NP}}
\newcommand{\NCo}{\textsc{NC}^1}

\newcommand{\Ppoly}{\textsc{P}/\text{poly}}
\newcommand{\poly}{\text{poly}}

\newcommand{\cli}{\textsc{Clique}}

\newcommand{\A}{\mathcal{A}}
\newcommand{\T}{\mathcal{T}}
\newcommand{\C}{\mathcal{C}}

\newcommand{\FO}{\text{FO}}
\newcommand{\PE}{\text{PE}}
\newcommand{\NDL}{\text{NDL}}

\newcommand{\ind}{\mathsf{ind}}
\newcommand{\tw}{\mathsf{tw}}

\renewcommand{\t}{\mathfrak{t}}
\newcommand{\tr}{\mathfrak{t}_\mathsf{r}}
\newcommand{\ti}{\mathfrak{t}_\mathsf{i}}
\newcommand{\qtw}{\q_{\mathsf{tw}}}

\newcommand{\s}{\mathfrak{s}}
\newcommand{\sr}{\mathfrak{s}_\mathsf{r}}
\newcommand{\si}{\mathfrak{s}_\mathsf{i}}

\newcommand{\NBP}{\textsc{NBP}}
\newcommand{\mNBP}{\textsc{NBP}_{\!\!\scriptscriptstyle+}}
\newcommand{\HGP}{\textsc{HGP}^2}
\newcommand{\mHGP}{\textsc{HGP}^2_{\!\!\scriptscriptstyle+}}
\newcommand{\HGPt}{\textsc{HGP}^3}
\newcommand{\mHGPt}{\textsc{HGP}^3_{\!\!\scriptscriptstyle+}}
\newcommand{\HGPn}{\textsc{HGP}^n}
\newcommand{\mHGPn}{\textsc{HGP}^n_{\!\!\scriptscriptstyle+}}
\newcommand{\NBC}{\textsc{NBC}}
\newcommand{\mNBC}{\textsc{NBC}_{\scriptscriptstyle+}}

\newcommand{\Cir}{\boldsymbol{C}}
\newcommand{\Pro}{H'}

\begin{document}

\title{On the Succinctness of Query Rewriting over \OWLQL{} Ontologies with Shallow  Chases}

\author{{S. Kikot$^1$, R. Kontchakov$^1$, V. Podolskii$^2$ and  M. Zakharyaschev$^1$}\\[8pt]
{\small\begin{tabular}{cc}
$^1$Department of Computer Science and Information Systems & $^2$Steklov Mathematical Institute\\
Birkbeck, University of London, U.K. & Moscow, Russia\\
\{kikot, roman, michael\}@dcs.bbk.ac.uk & podolskii@mi.ras.ru\end{tabular}}
}

\date{}

\maketitle

\begin{abstract}
We investigate the size of first-order rewritings of conjunctive queries over \OWLQL{} ontologies of depth 1 and 2 by means of hypergraph programs computing Boolean functions. Both positive and negative results are obtained. Conjunctive queries over ontologies of depth~1 have polynomial-size nonrecursive datalog rewritings; tree-shaped queries have polynomial positive existential rewritings; however, in the worst case, positive existential rewritings can only be of superpolynomial size. Positive existential and nonrecursive datalog rewritings of queries over ontologies of depth 2 suffer an exponential blowup in the worst case, while first-order rewritings are superpolynomial unless $\NP \subseteq \Ppoly$. We also analyse rewritings of tree-shaped queries over arbitrary ontologies and observe that the query entailment problem for such queries is fixed-parameter tractable.
\end{abstract}


\section{Introduction}\label{intro}

Our concern here is the size of conjunctive query (CQ) rewritings over \OWLQL{} ontologies. \OWLQL{}
({\small\url{www.w3.org/TR/owl2-profiles}})
is a profile of the Web Ontology Language \OWL{} designed for ontology-based data access (OBDA). In first-order logic, an \OWLQL{} ontology can be given as a finite set of sentences of the form
\begin{align}\label{tgd}
\forall \vec{x} \, \big (\varphi(\vec{x}) \to \exists \vec{y} \, \psi (\vec{x},\vec{y}) \big) \qquad \text{ or } \qquad
\forall \vec{x} \, \big (\varphi(\vec{x}) \land \varphi'(\vec{x}) \to \bot \big)
\end{align}
where $\varphi$, $\varphi'$ and $\psi$ are unary or binary predicates (such sentences are known as linear tuple-generating dependencies of arity 2 and disjointness constraints). \OWLQL{} is a (nearly) maximal fragment of \OWL{} enjoying first-order (FO) rewritability of CQs: given an ontology $\T$ and a CQ $\q(\vec{x})$, one can construct an  FO-formula $\q'(\vec{x})$ in the signature of $\q$ and $\T$ such that $\T,\A \models \q(\vec{a})$ iff $\A \models \q'(\vec{a})$, for any set $\A$ of ground atoms (data instance) and any tuple $\vec{a}$ of constants in $\A$. Thus, to find certain answers to $\q(\vec{x})$ over $(\T,\A)$, we can compute an FO-rewriting $\q'(\vec{x})$ and evaluate it over $\A$ using, for example, a database system.
The ontology $\T$ in the OBDA paradigm serves as a high-level global schema providing the user with a convenient query language over possibly heterogeneous data sources and enriching the data with additional knowledge. OBDA is widely regarded as a key to the new generation of information systems.
\OWLQL{} is based on the \mbox{\sl DL-Lite} family of description logics~\cite{CDLLR07,ACKZ09}; other languages supporting FO-rewritability of CQs include linear, sticky and sticky-join sets of tuple-generating dependencies~\cite{DBLP:journals/ai/CaliGP12, DBLP:conf/ijcai/BagetLMS09}.

In practice, rewriting-based OBDA systems\footnote{See, e.g., QuOnto~\cite{PLCD*08}, Presto/Prexto~\cite{DBLP:conf/kr/RosatiA10,DBLP:conf/esws/Rosati12}, Rapid~\cite{DBLP:conf/cade/ChortarasTS11}, Ontop~\cite{iswc13}, Requiem/Blackout~\cite{Perez-UrbinaMH09,Perez-Urbina12}, Nyaya~\cite{DBLP:conf/icde/GottlobOP11}, Clipper~\cite{DBLP:conf/aaai/EiterOSTX12} and \cite{DBLP:conf/rr/KonigLMT13}.} can only work efficiently with those CQs and ontologies that have \emph{reasonably short} rewritings. This obvious fact raises fundamental succinctness problems such as:
What is the size of FO-rewritings of CQs and \OWLQL{} ontologies in the worst case? Can rewritings of one type (say, nonrecursive datalog) be substantially shorter than rewritings of another type (say, positive existential)?
First answers to these questions were given in~\cite{icalp12} which constructed CQs $\q_n$ and ontologies $\T_n$, \mbox{$n < \omega$}, with only exponential positive existential (\PE) and nonrecursive datalog (\NDL) rewritings, and superpolynomial FO-rewritings (unless $\NP \subseteq \Ppoly$);~\cite{icalp12} also showed that \NDL-rewritings (FO-rewritings) can be exponentially (superpolynomially) more succinct than \PE-rewritings.
These prohibitively high lower bounds are caused by the fact that the chases (canonical models) for $\T_n$ contain full binary trees of depth $n$ and give rise to \emph{exponentially-many} homomorphisms from $\q_n$ to the labelled nulls of the chases, all of which have to be reflected in the rewritings of $\q_n$ and $\T_n$.

In this paper, we investigate succinctness of CQ rewritings over `shallow'  ontologies whose (polyno\-mial-size) chases are finite trees of depth 1 or 2 (which do not have chains of more than 1 or 2 labelled nulls). From the theoretical point of view, ontologies of depth 1 are important because their chases can only generate linearly-many homomorphisms of CQs to the labelled nulls; on the other hand, shallow ontologies are typical in the real-world OBDA applications. We obtain both positive and, unexpectedly, `negative' results summarised below:
\begin{enumerate}\itemsep=0pt
\item[(\emph{i})] any CQ and ontology of depth 1 have a polynomial-size \NDL-rewriting;

\item[(\emph{ii})] \PE-rewritings of some CQs and ontologies of depth 1 are of superpolynomial size;

\item[(\emph{iii})] any tree-shaped CQ and ontology of depth 1 have a  PE-rewriting of polynomial
size;

\item[(\emph{iv})] the existence of polynomial-size \FO-rewritings for all CQs and ontologies of depth 1 is equivalent to an open problem `$\NLogSpace/\poly \subseteq \NCo$?';

\item[(\emph{v})] \NDL- and \PE-rewritings of some CQs and  ontologies of depth 2 are of exponential size, while \FO-rewritings are of superpolynomial size unless $\NP \subseteq \Ppoly$.
\end{enumerate}
We prove (\emph{i})--(\emph{v}) by establishing a fundamental  connection between \FO-, \PE- and \NDL-rewritings, on the one hand, and, respectively, formulas,  monotone formulas and monotone circuits computing certain monotone Boolean functions, on the other.
These functions are associated with hypergraph representations of the tree-witness rewritings~\cite{DBLP:conf/kr/KikotKZ12}, reflecting  possible homomorphisms of the given CQ to the labelled nulls of the chases for the given ontology.
In particular, any hypergraph $H$ of degree~2 (every vertex in which belongs to 2 hyperedges) corresponds to a CQ $\q_H$ and an ontology $\T_H$ of depth 1 such that answering $\q_H$ over $\T_H$ and single-individual data instances amounts to computing the hypergraph function for $H$. We show that representing Boolean functions as hypergraphs of degree 2 is polynomially equivalent to representing their duals as nondeterministic branching programs (NBPs)~\cite{Jukna12}. This correspondence and known results on NBPs~\cite{Razborov91,Karchmer88} give (\emph{i}), (\emph{ii}) and (\emph{iv}) above. To prove (\emph{v}), we observe that hypergraphs of degree $3$ are computationally as powerful as nondeterministic Boolean circuits ($\textsc{NP}/\text{poly}$) and encode the function $\cli_{n,k}(\vec{e})$ (graph $\vec{e}$ with $n$ vertices has a $k$-clique) as  CQs over ontologies of depth~2. 
It also follows that there exist polynomial-size \FO-rewritings for all CQs and ontologies (of depth $2$) with polynomially-many tree witnesses iff  all functions in $\NP/\poly$ are com\-puted by polynomial-size formulas, that is, iff $\NP/\poly \subseteq \NCo$ (which is a well-known open problem).
Finally, we show that any tree-shaped CQ $\q$ and ontology $\T$ have a \PE-rewriting of size $O(|\T|^2 \cdot |\q|^{1 + \log d})$, where $d$ is a parameter related to the number of tree witnesses sharing a common variable. This gives (\emph{iii}) since $d = 2$ for ontologies of depth~1. We also note that the problem `$\T,\A \models \q$?'\!, for tree-shaped Boolean CQs and any $\T$, is fixed-parameter tractable (recall that the problem `$\A \models \q$?'\!, for tree-shaped $\q$, is known to be tractable~\cite{DBLP:conf/vldb/Yannakakis81}, while `$\T,\A \models \q$?' is \NP-hard~\cite{KKZ-DL11}).

As shown in~\cite{DBLP:conf/kr/GottlobS12}, exponential rewritings can be made polynomial at the expense of polynomially-many additional existential quantifiers over a domain with \emph{two constants} not necessarily occurring in the CQs; cf.~\cite{Avigad01}. Intuitively, \mbox{given $\q$,} $\T$ and $\A$, the extra quantifiers guess a homomorphism from $\q$ to the chase for $(\T,\A)$, whereas the standard rewritings (without extra constants) represent such homomorphisms explicitly (likewise NFAs are  exponentially more succinct than DFAs, and $\exists$-QBFs are exponentially more succinct than SAT).
A more practical utilisation of additional constants was suggested in the combined approach to OBDA~\cite{DBLP:conf/semweb/LutzSTW13}, where they are used to construct a polynomial-size encoding of the chase for the given ontology and data over which the original CQ is evaluated. This encoding may introduce (exponentially-many in the worst case) spurious answers that are eliminated by a special polynomial-time filtering procedure.


\section{The Tree-Witness Rewriting}\label{sec:tr-wit}

In this paper, we assume that an \emph{ontology}, $\T$, is a finite set of \emph{tuple-generating dependencies} (\emph{tgds}) of the form
\begin{align}\label{tgd1}
\forall \vec{x} \, \big( \varphi(\vec{x}) \to \exists \vec{y}  \bigwedge \psi_i (\vec{x},\vec{y}) \big),
\end{align}
where $\varphi$ and the $\psi_i$ are unary or binary atoms without constants and $|\vec{x} \cup \vec{y}| \le 2$. These tgds are expressible via tgds in~\eqref{tgd} using fresh binary predicates, whereas disjointness constraints in~\eqref{tgd} do not contribute to the size of  rewritings. Although the language given by~\eqref{tgd} is slightly different from  \OWLQL{}, all the results obtained here  are applicable to \OWLQL{} ontologies as well. When writing tgds, we will omit the universal quantifiers.  The size, $|\T|$, of $\T$ is the number of predicate occurrences in $\T$.
A \emph{data instance}, $\A$, is a finite set of ground atoms. The set of individual constants in $\A$ is denoted by $\ind(\A)$. Taken together, $\mathcal{T}$ and $\A$ form the \emph{knowledge base} (KB) $(\mathcal{T},\mathcal{A})$. To simplify notation, we will assume that the data instances in all KBs are \emph{complete} in the following sense: for any ground atom $S(\vec{a})$ with $\vec{a}\subseteq\ind(\A)$, if $\T,\A \models S(\vec{a})$ then $S(\vec{a}) \in \A$ (see Lemma~\ref{lemma:complete-data} below).

A \emph{conjunctive query} (CQ) $\q(\vec{x})$ is a formula $\exists \vec{y}\, \varphi(\vec{x}, \vec{y})$, where $\varphi$ is a conjunction of unary or binary atoms $S(\vec{z})$ with $\vec{z} \subseteq \vec{x} \cup \vec{y}$ (without loss of generality, we assume that CQs do not contain constants). A tuple $\vec{a}\subseteq \ind (\A)$ is a \emph{certain answer to $\q(\vec{x})$ over} $(\T,\A)$ if  $\mathcal{I} \models \q(\vec{a})$ for all models $\mathcal{I}$ of $\T$ and $\A$; in this case we write $\T,\A \models \q(\vec{a})$. If $\vec{x} = \emptyset$ then the CQ $\q$ is called \emph{Boolean}; a certain answer to such a $\q$ over $(\T,\A)$ is `yes' if $\T,\A \models \q$ and `no' otherwise. Where convenient, we regard a CQ as the set of its atoms.

Given a CQ $\q(\vec{x})$ and an ontology $\T$, an \FO-formula $\q'(\vec{x})$
\emph{without constants} is called an \emph{FO-rewriting of $\q(\vec{x})$ and $\T$} if, for any (complete) data instance $\A$ and any $\vec{a}\subseteq \ind(\A)$, we have \mbox{$(\T, \A) \models \q(\vec{a})$} iff $\A \models \q'(\vec{a})$.\!\footnote{Thus, we do not allow the  rewriting from~\cite{DBLP:conf/kr/GottlobS12} since it contains  constants.} If $\q'$ is a positive existential formula, we call it a \emph{PE-rewriting of $\q$ and $\T$}. We also consider rewritings in the form of nonrecursive datalog  queries.
Recall~\cite{Abitebouletal95} that a \emph{datalog program}, $\Pi$, is a finite set of Horn clauses
$\forall \vec{x}\, (\gamma_1 \land \dots \land \gamma_m \to \gamma_0)$,
where each $\gamma_i$ is an atom of the form $P(x_1,\dots,x_l)$ with  $x_i \in \vec{x}$. The atom $\gamma_0$ is the \emph{head} of the clause, and $\gamma_1,\dots,\gamma_m$ its \emph{body}. All variables in the head must also occur in the body.
A predicate $P$ \emph{depends} on $Q$ in $\Pi$ if $\Pi$ has a clause with $P$ in the head and $Q$ in the body; $\Pi$ is  \emph{nonrecursive} if this dependence relation is acyclic.
For a nonrecursive program $\Pi$ and an atom $\q'(\vec{x})$,
$(\Pi,\q')$ is called an \emph{\NDL-rewriting of $\q(\vec{x})$ and $\T$} in case $\T,\A\models \q(\vec{a})$ iff  $\Pi,\mathcal{A} \models \q'(\vec{a})$, for any (complete) $\A$ and $\vec{a} \subseteq \ind (\mathcal{A})$.
Rewritings over \emph{arbitrary} data are defined without stipulating that the data instances in KBs are complete.
\begin{lemma}\label{lemma:complete-data}
\textup{(}i\textup{)} For any  \textup{(}\PE-\textup{)} \FO-rewriting $\q'$ of $\q$ and $\T$ over complete data, there is a \textup{(}\PE-\textup{)} \FO-rewriting $\q''$ over arbitrary data with $|\q''| \leq O(|\q'| \cdot |\T|)$.

\textup{(}ii\textup{)} For any \NDL-rewriting $(\Pi,\q')$ of $\q$ and $\T$ over complete data, there is an \NDL-rewriting $(\Pi',\q')$ over arbitrary data with $|\Pi'| \leq |\Pi| + O(|\T|)$.
\end{lemma}
\begin{proof}
Given $\T$ and $\q$, we define a partial order $\leq$ on the formulas $\varrho(x)$ of the form $S(x)$, $S(x,x)$, $\exists y\,S(x,y)$ and  $\exists y\,S(y,x)$, where $S$ is a predicate in $\T$ or $\q$, to be the transitive and reflexive closure of the following relation $\le_1$:
\begin{equation*}
\varrho_1(x) \leq_1 \varrho_2(x)\quad\text{ iff }\quad \varrho_1(x) \to \varrho_2(x) \in \T.
\end{equation*}
This partial order $\leq$ is extended to the formulas $\varrho(x_1,x_2)$ of the form $S(x_1,x_2)$ and $S(x_2,x_1)$, where $S$ is a binary predicate in $\T$ or $\q$, by adding the transitive and reflexive closure of the following relation $\le_2$:
\begin{equation*}
\varrho_1(\vec{x}) \leq_2\varrho_2(\vec{x})\quad\text{ iff }\quad \varrho_1(\vec{x}) \to \varrho_2(\vec{x}) \in \T.
\end{equation*}
We also define an equivalence relation on such $\varrho(\vec{x})$ by taking $\varrho_1(\vec{x}) \equiv \varrho_2(\vec{x})$ iff $\varrho_1(\vec{x})\leq \varrho_2(\vec{x})$  and $\varrho_2(\vec{x})\leq \varrho_1(\vec{x})$. In each equivalence class, we fix a representative, denoted  $\varrho(\vec{x})/_\equiv$.

\smallskip

(\textit{i}) For \PE- and \FO-rewritings, we replace each $S(\vec{x})$ in  $\q'$ with a disjunction of all $\varrho(\vec{x})$ such that $\varrho(\vec{x}) \leq S(\vec{x})$. The size of the resulting rewriting $\q''$ increases linearly in $|\T|$.

\smallskip

(\textit{ii}) We assume without loss of generality that $\q'$ is not a predicate name in $\T$. Let $\Pi^*$ be the result of replacing each predicate name $S$ in $\Pi$ that occurs in $\T$  with a fresh predicate name $S^*$. Define $\Pi'$ to be the union of $\Pi^*$ and the following clauses:
\begin{align*}
S^*(\vec{x}) & \leftarrow \varrho(\vec{x}), && \text{ for all } \varrho(\vec{x}) \text{ with } \varrho(\vec{x}) \equiv S(\vec{x}),\\
S^*(\vec{x})  & \leftarrow \varrho^*(\vec{x}), && \text{ for all  } \varrho(\vec{x})/_\equiv \text{ such that } \varrho(\vec{x}) \text{ is an immediate predecessor of } S(\vec{x}) \text{ in } \leq,
\end{align*}
where $\varrho^*(\vec{x})$ is the result of replacing the predicate,  $S_1$, in $\varrho$ with $S_1^*$.
It should be clear that $(\Pi',\q')$ is an \NDL-rewriting of $\q$ and $\T$ over arbitrary data and that the size of the additional clauses in $\Pi'$ is linear in $|\T|$ (and does not depend on $\q$).
\end{proof}

We now define an improved version of the tree-witness PE-rewriting~\cite{DBLP:conf/kr/KikotKZ12} that will be used to establish links with formulas and circuits computing certain monotone Boolean functions.

As is well-known~\cite{Abitebouletal95}, for any KB $(\T,\A)$, there is a \emph{canonical model} (or \emph{chase}) $\mathcal{C}_{\T,\A}$ such that \mbox{$\T,\A \models \q(\vec{a})$} iff $\mathcal{C}_{\T,\A} \models \q(\vec{a})$, for all CQs $\q(\vec{x})$ and $\vec{a} \subseteq \ind(\A)$. The domain of $\C_{\T,\A}$ consists of $\ind(\A)$ and the witnesses, or \emph{labelled nulls}, introduced by the existential quantifiers in $\T$.

For any formula $\varrho(x)$ of the form $S(x)$, $S(x,x)$, $\exists y\, S(x,y)$ or $\exists y\, S(y,x)$, where $S$ is a  predicate in $\T$, we denote by $\C_\T^{\varrho(a)}$ the canonical model of the KB $(\T \cup \{ A(x) \to \varrho(x)\} , \{ A(a) \})$, where $A$ is a fresh unary predicate. We  say that $\T$ is \emph{of depth $k$}, $1 \le k < \omega$, if one of the $\C_\T^{\varrho(a)}$ contains a chain of the form $R_0(w_0,w_1)\dots R_{k-1}(w_{k-1},w_k)$, with not necessarily distinct $w_i$, but none of the $\C_\T^{\varrho(a)}$ has such a chain of greater length.

Suppose we are given a CQ $\q(\vec{x}) = \exists \vec{y}\, \varphi(\vec{x}, \vec{y})$ and an ontology $\T$. For a pair $\t = (\tr, \ti)$ of disjoint sets of variables in $\q$, with $\ti\subseteq \vec{y}$ and $\ti \ne\emptyset$ ($\tr$ can be empty), set
\begin{equation*}
\q_\t \ = \ \{\, S(\vec{z}) \in \q \mid \vec{z} \subseteq \tr\cup \ti \text{ and } \vec{z}\not\subseteq \tr\,\}.
\end{equation*}
We call $\t = (\tr, \ti)$ a \emph{tree witness for $\q$ and $\T$ generated by $\varrho$} if $\q_\t$ is a minimal subset of $\q$ for which there exists a homomorphism $h \colon \q_\t  \to \C_\T^{\varrho(a)}$ such that $\tr = h^{-1}(a)$ and $\q_\t$ contains all atoms of $\q$ with at least one variable from $\ti$ (cf.~aggregated unifiers from \cite{DBLP:conf/rr/KonigLMT13}).
Note that the same tree witness $\t = (\tr, \ti)$ can be generated by different $\varrho$.
Now, we set
\begin{equation}\label{tw-formula}
\tw_{\t}(\tr) \ \ =  \bigvee_{\t \text{ generated by } \varrho} \hspace*{-1em} \exists z\,\bigl(\varrho(z) \ \ \land \ \ \bigwedge_{x \in \tr} (x=z) \bigr).
\end{equation}
The variables in $\ti$ do not occur in $\tw_{\t}$ and are called \emph{internal}. The length, $|\tw_\t|$, of $\tw_\t$ is $O(|\q| \cdot |\T|)$.
Tree witnesses $\t$ and $\t'$ are \emph{conflicting} if $\q_\t \cap \q_{\t'} \ne \emptyset$. Denote by $\Theta^{\q}_{\T}$ the set of tree witnesses for $\q$ and $\T$. A subset $\Theta\subseteq \Theta^{\q}_{\T}$ is  \emph{independent} if no pair of distinct tree witnesses in it is conflicting. Let $\q_\Theta = \bigcup_{\t\in\Theta} \q_\t$.
The following PE-formula $\qtw$ is called the \emph{tree-witness rewriting of $\q$ and $\T$}:
\begin{equation}\label{rewriting0}
\qtw(\vec{x}) \ \ \ =  \hspace*{0em}  \bigvee_{\Theta \subseteq \Theta^{\q}_{\T} \text{ independent}} \hspace*{-0.4em} \exists\vec{y}\  \bigl(\hspace*{-1em}
\bigwedge_{S(\vec{z}) \in \q \setminus \q_\Theta}\hspace*{-1.2em} S(\vec{z})
 \ \land \ \bigwedge_{\t\in\Theta} \tw_\t(\tr) \,\bigr).
\end{equation}
\begin{example}\label{ex:conf}\em
Consider the following ontology and CQ:
\begin{align*}
\T & = \bigl\{ \  A_1(x) \to \exists y\, \bigl(R_1(x,y) \land Q(x,y)\bigr), \ \ A_2(x) \to \exists y\, \bigl(R_2(x,y) \land Q(y,x)\bigr) \ \bigr\}, \\
\q(x_1, x_2) & = \exists y_1y_2\, \bigl(R_1(x_1,y_1)\land Q(y_2,y_1)\land R_2(x_2,y_2)\bigr).
\end{align*}

\begin{figure}[ht]
\centerline{%
\begin{tikzpicture}[>=latex, point/.style={circle,draw=black,thick,minimum size=1.5mm,inner sep=0pt}, wiggly/.style={thick,decorate,decoration={snake,amplitude=0.3mm,segment length=2mm,post length=1mm}},
query/.style={thick},
tw/.style={shorten <= 0.1cm, shorten >= 0.1cm,dashed},yscale=1,xscale=0.9]\footnotesize
\coordinate (c1) at (0,0);
\coordinate (c2) at (1.5,1);
\coordinate (c3) at (3,0);
\coordinate (c4) at (4.5,1);
\node [fill=gray!2,thin,rounded corners,inner xsep=4mm,inner ysep=7mm,fit=(c2) (c3) (c4), label = below right : {}]{};
\node [draw,fill=gray!15,thin,rounded corners,inner xsep=4mm,inner ysep=5mm,fit=(c1) (c2) (c3), label = below right : {}]{};
\node [draw,thin,rounded corners,inner xsep=4mm,inner ysep=7mm,fit=(c2) (c3) (c4), label = below right : {}]{};
\draw[ultra thin, dashed] (-1.8,1) -- (6.3,1);
\draw[ultra thin, dashed] (-1.8,0) -- (6.3,0);
\node at (0,0.7) {\large $\t^1$};
\node at (4.5,0.3) {\large $\t^2$};
\node (t1) at (c1) [point, fill=black, label=below:{$x_1$}]{};
\node (t2) at (c2) [point, fill=white, label=above:{$y_1$}]{};
\node (t3) at (c3) [point, fill=white, label=below:{$y_2$}]{};
\node (t4) at (c4) [point, fill=black, label=above:{$x_2$}]{};
\draw[->,query] (t1)  to node [above, sloped]{$R_1$} (t2);
\draw[->,query] (t3)  to node [above, sloped]{$Q$} (t2);
\draw[->,query] (t4)  to node [pos=0.3, above, sloped]{$R_2$} (t3);
\node[fill=black] (c1) at (6.3,1) [point, label=above:{$A$}, label=right:{$a$}]{};
\node[fill=white] (c2) at (6.4,0) [point]{};
\draw[->,wiggly] (c1)  to node [below,sloped]{\scriptsize $R_2$} node [above,sloped]{\scriptsize $Q^-$} (c2);
\node at (6.5,-0.7) {$\C_{\T}^{A_2(a)}$};
\node[fill=black] (cc1) at (-1.8,0) [point, label=below:{$A$}, label=left:{$a$}]{};
\node[fill=white] (cc2) at (-1.8,1) [point]{};
\draw[->,wiggly] (cc1)  to node [below,sloped]{\scriptsize $R_1$} node [above,sloped]{\scriptsize $Q$} (cc2);
\node at (-1.6,1.7) {$\C_{\T}^{A_1(a)}$};
\end{tikzpicture}}%
\caption{Query $\q(x_1,x_2)$ and  canonical models $\C_{\T}^{A_1(a)}$ and $\C_{\T}^{A_2(a)}$ from Example~\ref{ex:conf}.}\label{fig:tws}
\end{figure}
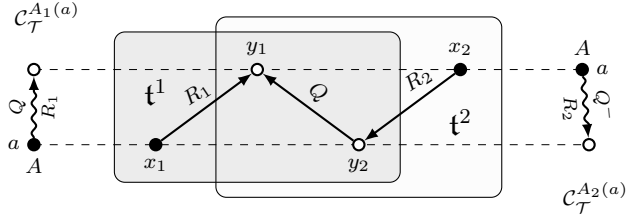

\noindent The CQ $\q$ is shown in Fig.~\ref{fig:tws} alongside the $\C_{\T}^{A_k(a)}$, $k = 1,2$.
There are two tree witnesses, $\t^1$ and $\t^2$, for $\q$ and $\T$ with
\begin{equation*}
\q_{\t^1} = \{\, R_1(x_1,y_1), Q(y_2,y_1)\,\}  \qquad\text{ and }\qquad  \q_{\t^2} = \{\, Q(y_2,y_1), R_2(x_2,y_2)\,\}
\end{equation*}
(they are shown as shaded rectangles in Fig.~\ref{fig:tws}); the tree witness $\t^k = (\tr^k, \ti^k)$, for $k = 1, 2$, is generated by $A_k(x)$
with $\tr^k = \{x_k,y_{3-k}\}$ and $\ti^k =\{y_k\}$,  which gives
\begin{equation*}
\tw_{\t^k}(x_k,y_{3-k})  \ \ = \ \ \exists z \, \bigl(A_k(z)\land(x_k = z) \land (y_{3-k} =z)\bigr).
\end{equation*}
As $\t^1$ and $\t^2$ are conflicting, we obtain the following rewriting:
\begin{equation*}
\exists y_1y_2 \, \big[\bigl(R_1(x_1,y_1) \land Q(y_2,y_1) \land R_2(x_2,y_2)\bigr) \ \ \lor\ \ 
\bigl(R_2(x_2,y_2) \land \tw_{\t^1}\bigr) \ \  \lor \ \ \bigl(R_1(x_1,y_1) \land \tw_{\t^2}\bigr) \big].
\end{equation*}
\end{example}
\begin{theorem}[\cite{DBLP:conf/kr/KikotKZ12}]
For any complete data instance $\A$ and any $\vec{a} \subseteq \ind(\A)$, we have $\T,\A \models \q(\vec{a})$ iff $\A \models \qtw(\vec{a})$.
\end{theorem}
The number of tree witnesses, $|\Theta^{\q}_{\T}|$, is bounded by $3^{|\q|}$. On the other hand, there is a sequence of queries $\q_n$ and ontologies $\T_n$ with exponentially many (in $|\q_n|$)
tree witnesses~\cite{DBLP:conf/kr/KikotKZ12}.
The length of $\qtw$ is $O(2^{|\Theta^{\q}_{\T}|} \cdot |\q| \cdot |\T|)$. If any two tree-witnesses for $\q$ and $\T$ are \emph{compatible}---that is, they are either non-conflicting or one is included in the other---then $\qtw$ can be equivalently transformed into the \PE-rewriting
\begin{equation*}
\qtw'(\vec{x}) \ \ \ = \ \ \exists \vec{y}\,\bigwedge_{S(\vec{z})\in \q} \bigl(\, S(\vec{z})  \ \ \lor \bigvee_{\t\in\Theta^{\q}_{\T} \text{ with } S(\vec{z})\in\q_{\t}} \hspace*{-2em}\tw_\t(\tr)\,\bigr)
\end{equation*}
of size $O(|\Theta^{\q}_\T| \cdot |\q|^2 \cdot |\T|)$.
Our aim now is to investigate transformations of this kind in the more abstract setting of Boolean functions. In Section~\ref{Representability:Degree:2}, we shall see an example of $\q$ and $\T$ with only $|\q|$-many tree witnesses any PE-rewriting of which is of superpolynomial size because of multiple combinations of incompatible tree witnesses.


\section{Hypergraph Functions}\label{sec:hyper-functions}

The rewriting $\qtw$ gives rise to monotone Boolean functions we call hypergraph functions.
For the complexity theory of monotone Boolean functions, the reader is referred to~\cite{Arora&Barak09,Jukna12}.
Let $H = (V,E)$ be a hypergraph with \emph{vertices} $v \in V$ and \emph{hyperedges} $e \in E$, $E \subseteq 2^V$. A subset $X \subseteq E$ is \emph{independent} if $e \cap e' = \emptyset$, for any distinct $e,e' \in X$. Denote by $V_X$ the set of vertices occurring in the hyperedges of $X$. With each $v \in V$ and $e \in E$ we associate propositional variables $p_v$ and $p_e$, respectively. The \emph{hypergraph function} $f_H$ for $H$ is given by the Boolean formula
\begin{equation}\label{hyper}
f_H \ \ = \ \bigvee_{
X \subseteq E \text{ independent}}
\Big( \bigwedge_{v \in V \setminus V_X}
\hspace*{-0.5em} p_v \ \land \ \bigwedge_{e \in X} p_e
\Big).
\end{equation}
The rewriting $\qtw$ of $\q$ and $\T$ defines a hypergraph $H^{\q}_\T$ whose vertices are the atoms of $\q$ and hyperedges are the sets $\q_\t$, for $\t\in\Theta^{\q}_{\T}$.
Formula~\eqref{hyper} for $H^{\q}_\T$ is the same as rewriting~\eqref{rewriting0} with the atoms $S(\vec{z})\in \q$ and the tree witness formulas $\tw_\t$, for $\t\in\Theta^\q_\T$, treated as propositional variables, $p_{S(\vec{z})}$ and $p_{\t}$, respectively.

\begin{example}\em
For $\q$ and $\T$ from Example~\ref{ex:conf}, the hypergraph $H^{\q}_\T$ is shown in Fig.~\ref{fig:ex:conf} and 
\begin{equation*}
f_{H^{\q}_\T} =  (p_{R_1(x_1,y_1)} \land p_{Q(y_2,y_1)} \land p_{R_2(x_2,y_2)}) \lor
(p_{R_2(x_2,y_2)} \land p_{\t^1}) \lor (p_{R_1(x_1,y_1)} \land p_{\t^2}).
\end{equation*}

\begin{figure}[ht]
\centerline{%
\begin{tikzpicture}[>=latex, point/.style={circle,draw=black,thick,minimum size=1.5mm,thick,inner sep=0pt},label distance=-2pt]
\coordinate (c1) at (0,0);
\coordinate (c2) at (1,0.7);
\coordinate (c3) at (2,0);
\fill[fill=gray!15,rounded corners=10] (-2.1,-0.2) -- (0.5,1.2) -- (2.7,1.2) -- (0.1,-0.2) -- cycle;
\draw[fill=gray!5,rounded corners=10] (4.1,-0.2) -- (1.5,1.2) -- (-0.7,1.2) -- (1.9,-0.2) -- cycle;
\draw[rounded corners=10] (-2.1,-0.2) -- (0.5,1.2) -- (2.7,1.2) -- (0.1,-0.2) -- cycle;
\node at (c1) [point,fill=white,label=left:{\scriptsize $R_1(x_1,y_1)$}] {};
\node at (c2) [point,fill=white,label=above:{\scriptsize $Q(y_2,y_1)$}] {};
\node at (c3) [point,fill=white,label=right:{\scriptsize $R_2(y_2,x_2)$}] {};
\node at (-0.1,0.5) {$\t^1$};
\node at (2.1,0.5) {$\t^2$};
\end{tikzpicture}}%
\caption{Hypergraph $H^\q_\T$ for $\q$ and $\T$ from Example~\ref{ex:conf}.}\label{fig:ex:conf}
\end{figure}
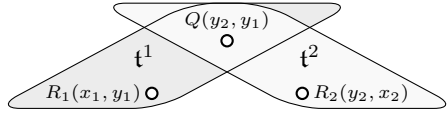
\end{example}

Suppose the function $f_{H^{\q}_\T}$ is computed by some Boolean formula $\chi$.  Consider the FO-formula obtained by adding the prefix $\exists \vec{y}$ to $\chi$ and replacing each $p_{S(\vec{z})}$ in it with $S(\vec{z})$ and each $p_{\t}$ with the formula $\tw_\t(\tr)$ of length $O(|\q| \cdot |\T|)$. By comparing~\eqref{hyper} and~\eqref{rewriting0}, we see that the resulting FO-formula is a  rewriting of $\q$ and $\T$. This gives the first claim in the following theorem; the second one requires some basic skills in datalog programming. (Recall~\cite{Arora&Barak09} that \emph{monotone} Boolean formulas and circuits contain only $\land$ and $\lor$.)

\begin{theorem}\label{NDL}
If $f_{H^{\q}_\T}$ is computed by a \textup{(}monotone\textup{)} Boolean formula $\chi$ then
there is a \textup{(}\PE-\textup{)} \FO-rewriting of $\q$ and $\T$ of size $O(|\chi| \cdot |\q| \cdot |\T|)$.

If $f_{H^{\q}_\T}$ is computed by a monotone Boolean circuit $\Cir$ then there is an \NDL-rewriting of $\q$ and $\T$ of size $O(|\Cir|\cdot |\q| \cdot |\T|)$.
\end{theorem}
\begin{proof}
The first claim was explained above, so we only prove the second one. Let $\q(\vec{x}) = \exists \vec{y}\,\varphi(\vec{x},\vec{y})$. First, we define a unary predicate $D_0$ by taking the following clauses
\begin{equation}\label{d1}
D_0(z) \leftarrow \varrho(z),
\end{equation}
for every formula $\varrho(z)$ of the form $S(z)$, $\exists y\,S(z,y)$ and $\exists y\,S(y,z)$, where $S$ is a predicate occurring in $\T$ or $\q$ (we say that such a $\varrho$ is \emph{in the signature of $\T$ and $\q$}). Intuitively, the interpretation of $D_0$ contains all the individual constants of the given data instance. We then set $\vec{z} = \vec{x} \cup\vec{y}$ and define a $|\vec{z}|$-ary predicate $D$ by the following clause
\begin{equation}
D(\vec{z}) \leftarrow \bigwedge_{z \in \vec{z}} D_0(z).
\end{equation}
We need the predicate $D$ to ensure that all the clauses in our \NDL-program are safe (that is, every variable in the head of a clause also occurs in the body).

Suppose $H^\q_\T$ has $m$ vertices and $l$ hyperedges.
Let $g_1, \dots, g_n$ be the nodes of $\Cir$ ordered in such a way that $g_1, \dots, g_m$ correspond to the atoms $S_1(\vec{z}_1), \dots, S_m(\vec{z}_m)$ of $\q$, $g_{m+1}, \dots, g_{m+l}$ correspond to the tree witnesses $\t^1, \dots, \t^l$ and $g_{m+l+1},\dots,g_n$ correspond to the gates of $\Cir$ with $g_n$ its output. For $1 \le i \le m$, we take the clauses
\begin{equation}\label{d2}
G_i(\vec{z}) \leftarrow S_i(\vec{z}_i) \land D(\vec{z}).
\end{equation}
For $m < i \le m+l$,
take the clauses
\begin{equation}\label{d3}
G_i(\vec{z}) \leftarrow \varrho(z_0) \land \bigwedge_{y \in \tr^{i-m}} (z_0=y) \land D(\vec{z}), \quad \text{ for all } \varrho(z) \text{ with } \t^{i-m} \text{ is generated by } \varrho(z).
\end{equation}
where $z_0$ is a fresh variable.
For $i > m+l$, we take the clauses
\begin{align}\label{d4}
G_i(\vec{z}) \leftarrow  G_j(\vec{z}) \land G_{j'}(\vec{z}) \land  D(\vec{z}),  & \quad\text{ if } g_i = g_j \land g_{j'},\\
\label{d6}
\left.\begin{array}{ll}%
G_i(\vec{z}) \leftarrow  G_j(\vec{z}) \land D(\vec{z}),\\[3pt]
G_i(\vec{z}) \leftarrow  G_{j'}(\vec{z}) \land D(\vec{z}),
\end{array}\right\}
 &\quad \text{ if } g_i = g_j \lor g_{j'}.
\end{align}
Denote the resulting set of clauses~\eqref{d1}--\eqref{d6} by $\Pi$.  We claim that $(\Pi, G_n)$ is an \NDL-rewriting of $\q$ and $\T$ over complete data.
To see this, we can transform $(\Pi, G_n)$ to a \PE-formula of the form
\begin{equation*}
\exists \vec{y} \, \Big[\psi(\vec{x}, \vec{y}) \land \bigwedge_{z \in \vec{x} \cup \vec{y}} \  \Bigl(
\bigvee_{\varrho \text{ in signature of $\T, \q$}} \hspace*{-1em}\varrho(z)
\Bigr)
 \Big],
\end{equation*}
where $\exists \vec{y} \,\psi(\vec{x}, \vec{y})$ can be constructed by taking the Boolean formula representing $\Cir$ and replacing $p_{S(\vec{z})}$ with $S(\vec{z})$ and $p_{\t}$ with $\tw_\t(\tr)$. It follows from the first claim of the theorem that $\exists \vec{y} \,\psi(\vec{x}, \vec{y})$ is a rewriting of $\q$ and $\T$ over complete data. It should be clear that the big conjunction does not change this fact.
\end{proof}

Thus, the problem of constructing short rewritings is reducible to the problem of finding short (monotone) Boolean formulas or circuits computing the hypergraph functions.

In the next section, we consider hypergraphs as programs for computing Boolean functions and compare them with the well-known  formalisms of nondeterministic branching programs (NBPs) and nondeterministic Boolean circuits~\cite{Arora&Barak09,Jukna12}.


\section{Hypergraphs, NBPs and Nondeterministic Boolean Circuits}\label{NBP}

Let $p_1,\dots,p_n$ be propositional variables. An \emph{input} to a hypergraph program or an NBP is a vector $\vec{\alpha}\in \{0,1\}^n$ assigning the truth-value $\vec{\alpha}(p_i)$ to each of the $p_i$.  We extend this notation to negated variables and constants by setting $\vec{\alpha}(\neg p_i) = \neg  \vec{\alpha}(p_i)$,  $\vec{\alpha}(0) = 0$ and $\vec{\alpha}(1) = 1$.
	
A \emph{hypergraph program} (HGP) is a hypergraph $H = (V,E)$ in which  every vertex is labelled with $0$, $1$, $p_i$ or $\neg p_i$. We say that the hypergraph program $H$ \emph{computes} a Boolean function $f$ in case, for any input $\vec{\alpha}$, we have $f(\vec{\alpha})=1$ iff there is an independent subset in $E$ that \emph{covers all zeros}---that is, contains all the vertices in $V$ labelled with 0 under $\vec{\alpha}$. A hypergraph program  is \emph{monotone} if there are no negated variables among its vertex labels. The \emph{size}, $|H|$, of a hypergraph program $H$ is the number of hyperedges in it.

We say that a hypergraph (program) $H$ is of \emph{degree} $\leq n$ if every vertex in it belongs to at most $n$ hyperedges; $H$ is of \emph{degree} $n$ if every vertex in it belongs to exactly $n$ hyperedges. We denote by $\text{HGP}(f)$ ($\HGPn(f)$) the minimal size of hypergraph programs (of degree $\leq n$) computing $f$; $\text{HGP}_{\!\!\scriptscriptstyle+}(f)$ and $\mHGPn(f)$ are used for the size of monotone programs.

We show first that monotone hypergraph programs of degree $\le 2$ capture the computational power of hypergraph functions for hypergraphs of degree $\le 2$. On the one hand, a monotone hypergraph program $H$ computes the  subfunction of $f_H$ obtained by setting $p_e = 1$, for all $e \in E$, and setting $p_v$ to be equal to the label of $v$.
On the other hand, any hypergraph function $f_H$ can be computed by a monotone hypergraph program of degree~2 and size $O(|H|)$.
\begin{lemma}\label{lemma:hyper:program}
For any hypergraph $H$ of degree~\mbox{$\le n$}, there is a monotone HGP of degree~$\le \max(2,n)$ and
size $2|H|$ computing the function $f_H$.
\end{lemma}
\begin{proof}
Given a hypergraph $H = (V,E)$, we label each $v \in V$ by a variable $p_v$. For each $e \in E$, we add  a fresh vertex $a_e$ labelled with $1$ and a fresh vertex $b_e$ labelled with $p_e$; then we create a new hyperedge $e' = \{a_e, b_e\}$ and add $a_e$ to the hyperedge $e$. We claim that the resulting hypergraph program $H'$
computes $f_H$. Indeed, for any input $\vec{\alpha}$ with $\vec{\alpha}(p_e) = 0$, we have to include the edge $e'$ into the cover, and so cannot include the edge $e$ itself. Thus, the program returns $1$ iff there is an independent set $X$ of hyperedges with $\vec{\alpha}(p_e)=1$, for all $e\in X$, covering all zeros of the variables $p_v$. It follows that $H'$ computes $f_H$.
\end{proof}
\begin{lemma}\label{lemma:hyper:degree2}
If $f$ is computable by a \textup{(}monotone\textup{)} HGP $H$ of degree $\leq 2$, then it can also be
computed by a \textup{(}monotone\textup{)} HGP  of degree~2 and size $|H| + 3$.
\end{lemma}
\begin{proof}
Let $v_1, \dots, v_k$ be vertices of degree $0$ and $v_{k+1}, \dots, v_l$ vertices of degree $1$ in $H$.
It suffices to extend $H$ with vertices, $x$, $y$, $z$ labelled with $1$, $0$, $0$, respectively, and  hyperedges $e_1 = \{v_1, \dots, v_l, x, y\}$, $e_2 = \{v_1, \dots, v_k, x, z\}$ and $e_3 = \{y, z\}$. 
It is easy to see that each cover should contain $e_3$ but cannot contain $e_1, e_2$.
Indeed, $y$ and $z$ should both be covered. However,  $e_1$ and $e_2$ intersect and cannot be both in the same cover. Thus, $y$ and $z$ should be covered by $e_3$, while $e_1$ and $e_2$, intersecting $e_3$, are not in the cover. After these choices we are left to deal with the original hypergraph.
Clearly, this construction preserves monotonicity.
\end{proof}

Our next result in this section establishes a link between  hypergraph programs of degree $\le 2$ and NBPs.
Recall~\cite{Jukna12} that an NBP is a directed multigraph with two distinguished vertices, $s$ and $t$, and the arcs labelled with 0, 1, $p_i$ or $\neg p_i$ (the arcs of the first type have no effect, the arcs of the second type are called \emph{rectifiers}, and those of the third and fourth types \emph{contacts}). We assume that $s$ has no incoming and $t$ no outgoing arcs, and note that NBPs may have multiple parallel arcs (with distinct labels) connecting two nodes.
We write $v \to_{\vec{\alpha}} v'$
if there is a directed path from $v$ to $v'$ labelled with 1 under  $\vec{\alpha}$.
An NBP{}  \emph{computes} a Boolean function $f$ if $f(\vec{\alpha}) = 1$ just in case $s\to_{\vec{\alpha}} t$. The \emph{size} of an NBP is the number of arcs in it.
An NBP is \emph{monotone} if it has no negated variables among its labels.
We denote by $\NBP(f)$ (respectively, $\mNBP(f)$) the minimal size of (monotone) NBPs computing $f$.
$\text{NBP}(\poly)$ is the class of Boolean functions computable by polynomial-size NBPs.
As usual, $f^*$ is the Boolean function dual to $f$.

\begin{theorem}\label{thm:deg_2}
{\rm (\emph{i})} For any Boolean function $f$,
$\HGP(f)$ and $\NBP(\neg f)$ are polynomially related.

\noindent {\rm (\emph{ii})} For any monotone Boolean function $f$,
$\mHGP(f)$ and $\mNBP(f^{\ast})$ are polynomially related.
\end{theorem}
\begin{proof}
We only prove {\rm (\emph{i})}; {\rm (\emph{ii})} is proved by the same argument.
Suppose $\neg f$ is computed by an NBP $G$. We construct a hypergraph program $H$ of degree $\leq 2$ as follows.
For each arc $e$ in $G$, $H$ has two vertices $e^0$ and $e^1$, which represent the beginning and the end of $e$. The vertex $e^0$ is labelled with the \emph{negated} label of $e$ in $G$ and $e^1$  with $1$. We also add to $H$ a vertex $t$ labelled with $0$.
For each arc $e$ in $G$, $H$ has an \emph{$e$-hyperedge} $\{e^0, e^1\}$.
For each vertex $v$ in $G$ but $s$ and $t$, $H$ has a \emph{$v$-hyperedge} that consists of all vertices $e^1$, for the arcs $e$ leading to $v$, and all vertices $e^0$, for the arcs $e$ leaving $v$. For the vertex $t$, $H$ contains a hyperedge that consists of $t$ and all vertices $e^1$, for the arcs $e$ leading to $t$.
We claim that the constructed hypergraph program $H$ computes $f$.
Indeed, if $s \not\to_{\vec{\alpha}} t$ in $G$ then the following subset of hyperedges is independent and covers all zeros: all $e$-hyperedges, for the arcs $e$ reachable from $s$ and labelled with 1 under $\vec{\alpha}$, and all $v$-hyperedges with $s\not\to_{\vec{\alpha}} v$.
Conversely, if $s\to_{\vec{\alpha}} t$ then it can be shown by induction that, for each arc $e_i$ of the path, the $e_i$-hyperedge must be in the cover of all zeros. Thus, no independent set can cover $t$, which is labelled with 0.

Suppose $f$ is computed by a hypergraph program $H$ of degree 2 with hyperedges $e_1, \dots, e_k$. We first provide a graph-theoretic characterisation of independent sets covering all zeros based on the implication graph~\cite{AspvallPlassTarjan79} (or the chain criterion of Lemma~8.3.1~\cite{Borgeretal97}). With any hyperedge $e_i$ we associate a propositional variable $p_{e_i}$ and with an input $\vec{\alpha}$ we associate the following set $\Phi_{\vec{\alpha}}$ of  binary clauses:
\begin{itemize}
\item[--] $\neg p_{e_i}  \lor \neg p_{e_j}$, if $e_i\cap e_j \ne \emptyset$ (informally: intersecting hyperedges cannot be chosen at the same time),
\item[--] $p_{e_i} \lor p_{e_j}$, if there is $v\in e_i\cap e_j$ such that $\vec{\alpha}(v) = 0$ (informally: all zeros must be covered; note that all vertices have at most two incident edges).
\end{itemize}
By definition, $X$ is an independent set covering all zeros iff  $X = \{ e_i \mid  \vec{\beta}(p_{e_i}) = 1\}$, for some assignment $\vec{\beta}$ satisfying $\Phi_{\vec{\alpha}}$. Let $B_{\vec{\alpha}} = (V, E_{\vec{\alpha}})$ be a directed graph with
\begin{align*}
V & ~=~ \bigl\{ e_i^+, e_i^- \mid 1\leq i \leq k\bigr\},\\
E_{\vec{\alpha}} & ~=~ \bigl\{ (e_i^+, e_j^-) \mid e_i \cap e_j \ne \emptyset \bigr\} \ \cup
\bigl\{ (e_i^-,e_j^+) \mid v\in e_i\cap e_j \text{ and } \vec{\alpha}(v) = 0 \bigr\}.
\end{align*}
($V$ is the set of all `literals' for the variables of $\Phi_{\vec{\alpha}}$ and $E_{\vec{\alpha}}$ is the arcs for the implicational form of the clauses of $\Phi_{\vec{\alpha}}$; note that $\neg p_{e_i} \lor \neg p_{e_j}$ gives rise to two implications, $p_{e_i} \to \neg p_{e_j}$ and $p_{e_j} \to \neg p_{e_i}$, and so to two arcs in the graph).
By~Lemma~8.3.1 in~\cite{Borgeretal97}, $\Phi_{\vec{\alpha}}$ is satisfiable iff there is no $e_i$ with a (directed) cycle going through $e^+_i$ and $e^-_i$.
It will be convenient for us to regard the $B_{\vec{\alpha}}$, for assignments $\vec{\alpha}$, as a single labelled directed graph $B$ with arcs of the from $(e_i^+, e_j^-)$ labelled with $1$ and arcs of the form $(e_i^-, e_j^+)$ labelled with $\neg v$, for $v\in e_i\cap e_j$. It should be clear that $B_{\vec{\alpha}}$ has a cycle going through $e_i^+$ and $e_i^-$ iff $e_i^- \to_{\vec{\alpha}} e_i^+$ and $e_i^+ \to_{\vec{\alpha}} e_i^-$ in $B$.

The required NBP will contain two distinguished vertices, $s$ and $t$, and, for each hyperedge $e_i$, two copies, $B_i^+$ and $B_i^-$, of $B$ with arcs from $s$ to the $e_i^-$ vertex of $B_i^+$, from the $e_i^+$ vertex of $B_i^+$ to the $e_i^+$ vertex of $B_i^-$ and from the $e_i^-$ vertex of $B_i^-$ to $t$. This construction  guarantees that $s\to_{\vec{\alpha}} t$ iff there is $e_i$ such that $B_{\vec{\alpha}}$ contains a cycle going through $e_i^+$ and $e_i^-$.
\end{proof}

In terms of expressive power, polynomial-size NBPs are a nonuniform analogue of the class $\NLogSpace$; in symbols: $\text{NBP}(\poly) = \NLogSpace/\poly$.
Compared to other nonuniform computational models, (monotone) NBPs sit between (monotone) Boolean formulas and Boolean circuits~\cite{Razborov91}.
As shown above, a (monotone) Boolean function $f$ is computable by a polynomial-size (monotone) HGP of degree $\le 2$ iff its dual $f^*$ is computable by a polynomial-size (monotone) NBP. (The problem whether $f^*$ can be replaced with $f$ is open; a negative solution would give a solution to the open problem~5 from~\cite{Razborov91}.)
Thus, (monotone) HGPs of degree $\le 2$ also sit between (monotone) Boolean formulas and Boolean circuits.
However, (monotone) hypergraphs of degree $\le 3$
turn out to be much more powerful than (monotone) hypergraphs of degree $\le 2$: we show now that polynomial-size (monotone) HGPs of degree $\le 3$ can compute \NP-hard Boolean functions.

A function $f \colon \{0,1\}^n \to \{0,1\}$ is computed by a \emph{nondeterministic Boolean circuit} $\Cir(\vec{x},\vec{y})$, with $|\vec{x}| = n$,
if for any $\vec{\alpha} \in \{0,1\}^n$, we have $f(\vec{\alpha})=1$ iff there is $\vec{\beta} \in \{0,1\}^m$ with $\Cir(\vec{\alpha},\vec{\beta})=1$.
The variables in $\vec{y}$ are called \emph{advice variables}.
We say that a nondeterministic circuit $\Cir(\vec{x},\vec{y})$ is \emph{monotone} if the negations in $\Cir$ are only applied to variables in $\vec{y}$. Denote by $\NBC(f)$ (respectively, $\mNBC(f)$) the minimal size of (monotone) nondeterministic Boolean circuits computing $f$.

\begin{theorem}\label{NBC}
\textup{(}i\textup{)} For any Boolean function $f$, $\textup{HGP}(f)$,  $\HGPt(f)$  and $\NBC(f)$ are polynomially related.

\noindent \textup{(}ii\textup{)} For any monotone Boolean function $f$, $\textup{HGP}_{\!\!\scriptscriptstyle +}(f)$, $\mHGPt(f)$ and $\mNBC(f)$ are polynomially related.
\end{theorem}
\begin{proof}
Clearly, $\text{HGP}(f) \leq \HGPt(f)$.

Now, given a  \textup{(}monotone\textup{)}  HGP of size $m$, we construct a \textup{(}monotone\textup{)}  nondeterministic circuit $\Cir(\vec{x},\vec{y})$ of size $\poly(m)$.
Its $\vec{x}$-variables are the variables of the program, and its advice variables correspond to the edges of the program.
The circuit $\Cir$ will return 1 on $(\vec{\alpha},\vec{\beta})$ iff the family $\{e_i \mid \vec{\beta}(e_i) = 1\}$ of edges of the hypergraph
forms an independent set covering all zeros under $\vec{\alpha}$.
It is easy to construct a polynomial-size circuit checking this property.
Indeed, for each pair of intersecting edges $e_i, e_j$, it is enough  to take disjunction $\neg e_i \vee \neg e_j$, 
and for each vertex of the hypergraph labelled with $p$ and adjacent to edges $e_{i_1}, \ldots, e_{i_k}$
to take disjunction $p \vee  e_{i_1} \lor \dots \lor e_{i_k}$. (Note that applications of $\neg$ to advice variables in the monotone case are allowed.)
It then remains to take a conjunction of these disjunctions.
Finally, it is easy to see that the resulting nondeterministic circuit is monotone if the hypergraph program is monotone.

Conversely, suppose $f$ is computed by a nondeterministic circuit $\Cir(\vec{x},\vec{y})$. Let $g_1,\dots,g_n$ be the nodes of $\Cir$ (including the inputs $\vec{x}$ and $\vec{y}$). We construct an HGP of degree $\le 3$ computing $f$ by taking, for each $i$, a vertex $g_i$ labelled with $0$ and a pair of hyperedges $\bar{e}_{g_i}$ and $e_{g_i}$, both containing $g_i$. No other edge contains $g_i$, and so either $\bar{e}_{g_i}$ or $e_{g_i}$ should be present in any cover of zeros. (Intuitively, if the node $g_i$ is positive then $e_{g_i}$ belongs to the cover; otherwise, $\bar{e}_{g_i}$ is there.) To ensure this property, for each input variable $x_i$, we add a vertex labelled with $\neg x_i$ to $e_{x_i}$  and a fresh vertex labelled with $x_i$ to $\bar{e}_{x_i}$.
For each gate $g_i$, we consider three cases.
\begin{itemize}
\item[--] If $g_i = \neg g_j$ then we add a vertex labelled with $1$ to $e_{g_i}$ and $\bar{e}_{g_j}$, and a vertex labelled with $1$ to $\bar{e}_{g_i}$ and $e_{g_j}$.

\item[--] If $g_i = g_j \vee g_{j'}$ then we add a vertex labelled with $1$ to $e_{g_j}$ and $\bar{e}_{g_i}$,
add a vertex labelled with $1$ to $e_{g_{j'}}$ and $\bar{e}_{g_i}$;
then, we add vertices $h_j$ and $h_{j'}$ labelled with $1$ to $\bar{e}_{g_j}$ and $\bar{e}_{g_{j'}}$, respectively, and a vertex $u_{i}$ labeled with $0$ to $\bar{e}_{g_i}$; finally, we add hyperedges $\{h_j, u_i\}$ and $\{h_{j'}, u_i\}$.

\item[--] If $g_i = g_j \wedge g_{j'}$ then we use the dual  construction.
\end{itemize}
It is not hard to see that $e_{g_i}$ is in the cover iff it contains $\bar{e}_{g_j}$ in the first case, and $e_{g_i}$ is in the cover iff it contains at least one of $e_{g_j}$ and $e_{g_{j'}}$ in the second one.
Indeed, in the second case if, say, the cover contains $e_{g_j}$  then it cannot contain $\bar{e}_{g_i}$, and so it contains $e_{g_i}$. The vertex $u_i$ in this case can be covered by the hyperedge $\{h_j, u_i\}$ since $\bar{e}_{g_j}$ is not in the cover.
Conversely, if neither $e_{g_j}$ nor $e_{g_{j'}}$ is in the cover, then it must contain both $\bar{e}_{g_j}$ and
$\bar{e}_{g_{j'}}$ and so, neither $\{h_j, u_i\}$ nor $\{h_{j'}, u_i\}$ can belong to the cover
and we will have to include $\bar{e}_{g_i}$ to the cover.
Finally, we add one more vertex labelled with $0$ to $e_{g}$ for the output gate $g$ of $\Cir$. By induction on the structure of $\Cir$ one can show that, for each $\vec{\alpha}$, there is $\vec{\beta}$ such that $\Cir(\vec{\alpha},\vec{\beta})=1$ iff the constructed HGP returns 1 on $\vec{\alpha}$.

If $\Cir$ is monotone, we remove all vertices labelled with $\neg x_i$. Then, for an input $\vec{\alpha}$, there is a cover of zeros  in the resulting HGP iff there are $\vec{\beta}$ and $\vec{\alpha}' \leq \vec{\alpha}$ with $\Cir(\vec{\alpha}',\vec{\beta})=1$.
\end{proof}

Now, we use the developed machinery to investigate the size of rewritings over ontologies of depth~1 and~2.


\section{Rewritings over Ontologies of Depth 1}\label{Representability:Degree:2}

\begin{theorem}\label{depth1}
For any ontology $\T$ of depth~$1$ and any CQ $\q$, the hypergraph $H^{\q}_{\T}$ is of degree $\leq 2$ and $|\Theta_{\T}^{\q}| \leq |\q|$.
\end{theorem}
\begin{proof}
We have to show that every atom in $\q$ belongs to at most two $\q_\t$, $\t\in \Theta_{\T}^{\q}$.
Suppose $\t = (\tr,\ti)$ is a tree witness and $y \in \ti$. Since $\T$ is of depth~1, $\ti = \{y\}$ and $\tr$ consists of all those variables $z$ in $\q$ for which $S(y,z) \in \q$ or $S(z,y) \in \q$, for some $S$.  Thus, different tree witnesses have different internal variables $y$. An atom of the form $A(u)\in\q$ is in $\q_\t$ iff $u = y$. An atom of the form $P(u,v)\in\q$ is in $\q_\t$ iff either $u = y$ or $v=y$. Therefore, $P(u,v)\in\q$ can only be covered by the tree witness with internal $u$ and by the tree witness with internal $v$.
\end{proof}

\begin{theorem}\label{polyNDL}
Any CQ $\q$ and ontology $\T$ of depth $1$ have a polynomial-size $\NDL$-rewriting.
\end{theorem}
\begin{proof}
By Theorem~\ref{depth1},  the hypergraph $H^\q_\T$ is of degree $\leq 2$, and so, by Lemma~\ref{lemma:hyper:program}, there is a polynomial-size HGP of degree~$\leq 2$ computing $f_{H^\q_\T}$. By Theorem~\ref{thm:deg_2}, we have a polynomial-size monotone NBP computing $f^*_{H^\q_\T}$. But then we also have a polynomial-size monotone Boolean circuit that  computes $f^*_{H^\q_\T}$ (see, e.g., \cite{Razborov91}). By swapping $\land$ and $\lor$ in this circuit, we obtain a polynomial-size monotone circuit computing $f_{H^\q_\T}$. It remains to apply Theorem~\ref{NDL}.
\end{proof}

We show next that any hypergraph $H$ of degree 2 is representable by means of a CQ $\q_H$ and an ontology $\T_H$ of depth~1 in the sense that $H$ is isomorphic to $H^{\q_H}_{\T_H}$.
We can assume that $H = (V, E)$ comes with two fixed maps $i_1,i_2\colon V \to E$ such that $i_1(v) \ne i_2(v)$, $v \in i_1(v)$ and $v \in i_2(v)$, for any  $v\in V$.
For each hyperedge $e\in E$, we take an individual variable $z_e$ and let $\vec{z}$ be the vector of these variables. For every vertex $v \in V$, we take a binary predicate $R_v$ and set:
\begin{equation*}
\q_H ~=~ \exists \vec{z}\, \bigwedge_{v\in V} R_v(z_{i_1(v)}, z_{i_2(v)}).
\end{equation*}
Let $\T_H$ be an ontology with the following tgds, for $e \in E$:
\begin{equation}\label{eq:TH:exists}
 A_e(x) \ \ \to \ \ \exists y\, \Bigl[ \bigwedge_{\begin{subarray}{c}v\in V\\i_1(v) = e\end{subarray}} R_v(y,x) \ \ \land \bigwedge_{\begin{subarray}{c}v\in V\\i_2(v) = e\end{subarray}} R_v(x,y)\Bigr].
\end{equation}
\begin{example}\label{example1}\em
Consider $H = (V, E)$ with $V = \{v_1, v_2, v_3, v_4\}$ and $E = \{e_1, e_2, e_3\}$, where 
\begin{equation*}
e_1 = \{v_1, v_2, v_3\},\qquad e_2 = \{v_3, v_4\},\qquad e_3 = \{v_1, v_2, v_4\}, 
\end{equation*}
and assume that
\begin{align*}
& i_1\colon v_1\mapsto e_1, \ \ v_2\mapsto e_3, \ \ v_3 \mapsto e_1, \ \ v_4 \mapsto e_2,\\
& i_2\colon v_1\mapsto e_3, \ \ v_2 \mapsto e_1, \ \ v_3 \mapsto e_2, \ \  v_4 \mapsto e_3.
\end{align*}
The hypergraph $H$ is shown in Fig.~\ref{fig:hyper-witness}, where each $v_k$ is represented by an edge, $i_1(v_k)$ is indicated by the circle-shaped end of the edge and $i_2(v_k)$ by the diamond-shaped end of the edge; the $e_j$ are shown as large grey squares.

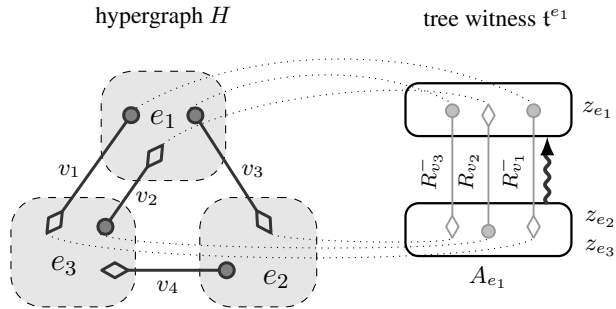
\begin{figure}[ht]
\centerline{%
\begin{tikzpicture}[>=latex, tpoint/.style={circle,draw=black,line width=0.7mm,minimum size=1.5mm,inner sep=0pt},query/.style={open diamond-*,line width=0.4mm,fill=gray,draw=black!80},tquery/.style={line width=0.25mm,fill=gray!50,draw=black!40},scale=1.2]
\coordinate (v12) at (1.7,0.3); 
\coordinate (v22) at (3,1.3); 
\coordinate (v42) at (2.3,-0.1); 
\coordinate (v11) at (2.7,1.7); 
\coordinate (v21) at (2.3,0.3); 
\coordinate (v31) at (3.3,1.7); 
\coordinate (v32) at (4.2,0.3); 
\coordinate (v41) at (3.8,-0.1); 
\coordinate (v61) at (4.3,-0.1); 
\node[rectangle,draw=black,fill=gray!20,dashed,rounded corners=12,fit=(v41) (v61) (v32),inner sep=15]  (e2) {\large \phantom{$e_2$}};
\node at ($(e2)+(0.2,-0.3)$) {\large $e_2$};
\node[rectangle,draw=black,fill=gray!20,dashed,rounded corners=12,fit=(v11) (v22) (v31),inner sep=15] (e1) {\large $e_1$};
\node[rectangle,draw=black,fill=gray!20,dashed,rounded corners=12,fit=(v12) (v21) (v42),inner sep=15] (e3) {\large \phantom{$e_3$}};
\node at ($(e3)+(-0.1,-0.15)$) {\large $e_3$};
\draw[->,query] (v12) -- (v11) node[left, midway] {\small $v_1$};
\draw[->,query] (v22) -- (v21)  node[right, pos=0.6] {\small $v_2$};
\draw[->,query] (v42) -- (v41) node[below, midway] {\small $v_4$};
\draw[->,query] (v32) -- (v31) node[right, midway] {\small $v_3$};
\node[rotate=0] at (3,2.7) {\small hypergraph $H$};
\node[rotate=0] at (6.7,2.7) {\small tree witness $\t^{e_1}$};
\coordinate (v31c) at (6.2,1.75);
\coordinate (v22c) at (6.6,1.75);
\coordinate (v11c) at (7.1,1.75);
\coordinate (w1c) at (7,1.68);
\coordinate (w1c0) at (6.2,1.68);
\coordinate (v32c) at (6.2,0.25);
\coordinate (v21c) at (6.6,0.25);
\coordinate (v12c) at (7.1,0.25);
\coordinate (a1c) at (7,0.35);
\coordinate (a1c0) at (6.2,0.35);
\node[rectangle,draw=black,line width=0.3mm,rounded corners=5,fit=(w1c0) (w1c),inner xsep=20, inner ysep=10,label=right:{\small $z_{e_1}$}] (w1) {};
\node[rectangle,draw=black,line width=0.3mm,rounded corners=5,fit=(a1c0) (a1c),inner xsep=20, inner ysep=10,label=right:{\hspace*{-0.5em}\small\begin{tabular}{c}$z_{e_2}$\\$z_{e_3}$\end{tabular}},label=below:{\small $A_{e_1}$}] (a1) {};
\draw[->,line width=0.5mm,draw=black!80,decorate,decoration={snake,amplitude=0.3mm,segment length=2mm,post length=1mm}] ($(a1)+(0.65,0.28)$) -- ($(w1)+(0.65,-0.3)$) node[right, midway] {}; 
\draw[open diamond-*,tquery]  (v32c) -- (v31c) node[above, midway,sloped] {\footnotesize $R_{v_3}^-$};
\draw[*-open diamond,tquery] (v21c) -- (v22c)  node[above, midway,sloped] {\footnotesize $R_{v_2}$};
\draw[open diamond-*,tquery] (v12c) -- (v11c) node[above, midway,sloped] {\footnotesize $R_{v_1}^-$};
\draw[dotted,bend left,looseness=0.8] (v31) to  (v31c);
\draw[dotted,bend left,looseness=0.6] (v22) to  (v22c);
\draw[dotted,bend left,looseness=0.8] (v11) to  (v11c);
\draw[dotted,bend right,looseness=0.2] (v32) to  (v32c); %
\draw[dotted,bend right,looseness=0.2] (v21) to  (v21c);
\draw[dotted,bend right,looseness=0.3] (v12) to  (v12c);
\end{tikzpicture}}%
\caption{A hypergraph $H$ and a tree witness for $\q_H$ and $\T_H$.}\label{fig:hyper-witness}
\end{figure}

\noindent In this case,
\begin{equation*}
\q_H = \exists z_{e_1}z_{e_2}z_{e_3} \,\bigl(R_{v_1}(z_{e_1},z_{e_3}) \land R_{v_2}(z_{e_3},z_{e_1}) \land
R_{v_3}(z_{e_1},z_{e_2})\land R_{v_4}(z_{e_2},z_{e_3}) \bigr)
\end{equation*}
and the ontology $\T_H$ consists of the following tgds:
\begin{align*}
A_{e_1}(x) &\to \exists y\,  \bigl[R_{v_1}(y,x)\land R_{v_2}(x,y)\land R_{v_3}(y,x)\bigr],\\
A_{e_2}(x) &\to \exists y\, \bigl[  R_{v_3}(x,y) \land  R_{v_4}(y,x)\bigr], \\
A_{e_3}(x) &\to \exists y\,  \bigl[ R_{v_1}(x,y) \land R_{v_2}(y,x)\land R_{v_4}(x,y)\bigr].
\end{align*}
The canonical model $\C^{A_{e_1}(a)}_{\T_H}$ is shown on the right-hand side of the picture above. Note that each $z_{e}$ determines the tree witness $\t^{e}$ with $\q_{\t^e} = \{R_v(z_{i_1(v)}, z_{i_2(v)}) \mid v \in e\}$; $\t^e$ and $\t^{e'}$ are conflicting iff $e \cap e' \ne \emptyset$. It follows that $H$ is isomorphic to $H_{\T_H}^{\q_H}$. In fact, this example generalises to the following:
\end{example}
\begin{theorem}\label{representable}
Any hypergraph $H$ of degree~2 is isomorphic to $H^{\q_H}_{\T_H}$, with $\T_H$ being an ontology of depth $1$.
\end{theorem}
\begin{proof}
We show that the map $h\colon v\mapsto R_v(z_{i_1(v)}, z_{i_2(v)})$ is an isomorphism between $H$ and $H^{\q_H}_{\T_H}$. By the definition of $\q_H$, $h$ is a bijection between $V$ and the atoms of $\q_H$. For any $e\in E$, there is a tree witness
$\t^e = (\tr^e, \ti^e)$ generated by $A_e(x)$ with
\begin{equation*}
\ti^e = \{z_e\} \ \ \text{ and  } \ \
\tr^e = \{z_{e'} \mid e' \cap e \neq\emptyset\},
\end{equation*}
and $\q_{\t^e}$ consists of the $h(v)$, for $v\in e$. Conversely, every tree witness $\t$ for $\q_H$ and $\T_H$ contains $z_e \in \ti$, for some $e\in E$, and so
$\q_\t = \{ h(v) \mid v \in e\}$.
\end{proof}

We now show that answering $\q_H$ over $\T_H$ and certain single-individual data instances amounts to computing the Boolean function $f_H$.
Let $H = (V,E)$ be a hypergraph of degree~2 with $V = \{v_1,\dots,v_n\}$ and $E  = \{e_1,\dots,e_m\}$. We denote by $\vec{\alpha}(v_i)$ the $i$-th component of $\vec{\alpha}\in \{0,1\}^n$, by $\vec{\beta}(e_j)$ the $j$-th component of $\vec{\beta}\in\{0,1\}^m$, and set
\begin{equation*}
\A_{\vec{\alpha},\vec{\beta}} \  = \  \{\,R_{v_i}(a,a) \mid \vec{\alpha}(v_i) = 1\,\} \ \cup \ \{\,A_{e_j}(a) \mid \vec{\beta}(e_j) = 1\,\}.
\end{equation*}

\begin{theorem}\label{p2}
Let $H = (V,E)$ be a hypergraph of degree~$2$. Then $\T_H, \A_{\vec{\alpha},\vec{\beta}} \models \q_H$ iff $f_H(\vec{\alpha},\vec{\beta}) = 1$, for any $\vec{\alpha}\in \{0,1\}^{|V|}$ and $\vec{\beta}\in \{0,1\}^{|E|}$.
\end{theorem}
\begin{proof}
$(\Leftarrow)$ Let $X$ be an independent subset of $E$ such that $\bigwedge_{v \in V \setminus V_X} p_v \land \bigwedge_{e \in X} p_e$ is true on $\vec{\alpha}$ (for the $p_v$) and $\vec{\beta}$ (for the $p_e$).
Define $h \colon \q_H \to \C_{\T_H, \A_{\vec{\alpha},\vec{\beta}}}$
by taking $h(z_e) = a$ if $e\notin X$ and  \mbox{$h(z_e) = w_e$}  otherwise,
where $w_e$ is the labelled null in the canonical model $\C_{\T_H, \A_{\vec{\alpha},\vec{\beta}}}$ introduced to witness the existential quantifier in~\eqref{eq:TH:exists}.
One can check that $h$ is a homomorphism, and so $\T_H, \A_{\vec{\alpha},\vec{\beta}} \models \q_H$.

\smallskip

\noindent $(\Rightarrow)$ Suppose $h \colon \q_H \to \C_{\T_H, \A_{\vec{\alpha}, \vec{\beta}}}$ is a homomorphism. We show that the set $X = \{e \in E  \mid h(z_e) \neq a\}$ is independent. Indeed, if $e, e' \in X$ and $v \in e \cap e'$, then $h$ sends one variable of the $R_v$-atom to the labelled null $w_e$ and the other end to  $w_{e'}$, which is impossible.  We claim  that $f_H(\vec{\alpha},\vec{\beta}) = 1$. Indeed, for each $v \in V\setminus V_X$, $h$ sends both ends of the $R_v$-atom to $a$, and so $\vec{\alpha}(v) = 1$. For each $e \in X$, we must have $h(z_e) = w_{e}$ because $h(z_e) \neq a$, and so $\vec{\beta}(e) = 1$. It follows that $f_H(\vec{\alpha}, \vec{\beta}) = 1$.
\end{proof}

We are now fully equipped to show that there exist CQs and ontologies of depth 1 without polynomial-size \PE-rewritings:
\begin{theorem}\label{Depth1:PE}
There is a sequence of CQs $\q_n$ and ontologies $\T_n$ of depth $1$, both of polynomial size in $n$, such that any \PE-rewriting of $\q_n$ and $\T_n$ is of size $n^{\Omega(\log n)}$. 
\end{theorem}
\begin{proof}
As shown in~\cite{Karchmer88}, there is a sequence $f_n$ of monotone  Boolean functions that are computable by polynomial-size  monotone NBPs, but any monotone Boolean formulas computing $f_n$ are of size  $n^{\Omega(\log n)}$. In fact, $f_n$ from~\cite{Karchmer88} checks whether two given vertices are connected by a path in a given undirected graph (alternatively, one could use the functions from~\cite{GrigniSipser95}). 

By Theorem~\ref{thm:deg_2}~(\emph{ii}) and Lemma~\ref{lemma:hyper:degree2}, there is a sequence of polynomial-size monotone HGPs $\Pro_n$  of degree~$2$ computing $f_n^{*}$. By applying Theorem~\ref{representable} to the hypergraph $H_n$ of $\Pro_n$, we obtain a sequence of CQs $\q_n$ and ontologies $\T_n$ of depth 1 such that $H_n$ is isomorphic to $H^{\q_n}_{\T_n}$.  We show now that any \PE-rewriting $\q'_n$ of $\q_n$ and $\T_n$ can be transformed to a monotone Boolean formula computing $f_n$ and having size $\le |\q'_n|$.

To define such a formula, we eliminate the quantifiers in $\q'_n$ in the following way:  take a constant $a$ and replace every subformula of the form $\exists x\,\psi(x)$ in $\q'_n$ with $\psi(a)$, repeating this operation as many times as possible. The resulting formula $\q''_n$ is built from atoms of the form $A_e(a)$, $R_v(a,a)$ and $S_e(a,a)$ using $\land$ and $\lor$. For every data instance $\A$ with a single individual $a$, we have $\T_n,\A \models \q_n$ iff $\A \models \q_n''$. Let $\chi_n$ be the result of replacing $S_e(a,a)$ in $\q''_n$ with $\bot$, $A_e(a)$ with $p_e$ and $R_v(a,a)$ with $p_v$. Clearly, $|\chi_n| \le |\q'_n|$.
By the definition of $\A_{\vec{\alpha}, \vec{\beta}}$ and Theorem~\ref{p2}, we obtain:
\begin{equation*}
\chi_n(\vec{\alpha}, \vec{\beta}) = 1 \quad \text{iff} \quad  \A_{\vec{\alpha}, \vec{\beta}} \models \q''_n \quad \text{iff} \quad  
\T_n,\A_{\vec{\alpha}, \vec{\beta}} \models \q_n \quad \text{iff} \quad  f_{H_n}(\vec{\alpha}, \vec{\beta}) = 1.
\end{equation*}
As $\Pro_n$ computes $f_n^{*}$, we can obtain $f_n^*$ from $f_{H_n}$ by replacing each  $p_e$ with 1 and each $p_v$ with the label of $v$ in $\Pro_n$. The same substitution in $\chi_n$ (with $\top$ and $\bot$ in place of 1 and 0) gives a monotone formula that computes $f_n^{*}$. By swapping $\lor$ and $\land$ in it, we obtain a monotone formula $\chi'_n$ computing $f_n$. It remains to recall that $|\q'_n| \ge |\chi'_n| = n^{\Omega(\log n)}$.
\end{proof}

It may be of interest to note that the function $f_n$ in the proof above is in the complexity class $\LogSpace$. The algorithm computing $f_n$ by querying the \NDL-rewriting of Theorem~\ref{polyNDL} over single-individual data instances runs in polynomial time; the algorithm querying any \PE-rewriting to compute $f_n$ requires, by Theorem~\ref{Depth1:PE}, superpolynomial time.

We note further that instead of reachability in undirected graphs in Theorem~\ref{Depth1:PE} we could use reachability in directed graphs. Indeed, since the undirected case reduces to the directed one, we have the same lower bound for computing directed reachability by monotone formulas.
On the other hand, it is known that directed reachability also can be computed by polynomial-size monotone circuits.
As reachability in directed graphs is $\NLogSpace/\poly$-complete under $\NCo$-reductions, the argument in the proof of Theorem~\ref{Depth1:PE} shows that
the existence of short \FO-rewritings of CQs and ontologies of depth~1 is equivalent to a well-known open problem in computational complexity:
\begin{theorem}
There exist polynomial-size \FO-rewritings for all CQs and ontologies of depth $1$ iff  all functions in  $\NLogSpace/\poly$ are computed by polynomial-size Boolean formulas, that is, iff $\NLogSpace/\poly \subseteq \NCo$.
\end{theorem}
\begin{proof}
$(\Leftarrow)$ Suppose $\NLogSpace/\poly \subseteq \NCo$. Take an arbitrary CQ and an ontology of depth $1$. By Theorem~\ref{depth1}, its hypergraph $H_\T^\q$ is of degree~$\leq 2$ and polynomial size. By Lemma~\ref{lemma:hyper:program}, there is a polynomial-size HGP $H$ computing $f_{H_\T^\q}$, whence, by Theorem~\ref{thm:deg_2}~(\emph{i}), there is a polynomial-size NBP computing $\neg f_{H_\T^\q}$, and so $f_{H_\T^\q}$ is in $\coNLogSpace/\poly = \NLogSpace/\poly$.
Therefore, by our assumption,  it can be computed by a polynomial-size Boolean formula. 
By Theorem~\ref{NDL}, the latter translates into a polynomial-size \FO-rewriting of $\q$ and $\T$.

\smallskip

\noindent $(\Rightarrow)$ Suppose that there exist polynomial-size \FO-rewritings  for all CQs and ontologies of depth~$1$. Consider a sequence of functions $f_n$ that compute the connectivity function in \emph{directed} graphs. Since $f_n \in \NLogSpace$ and $\NLogSpace = \coNLogSpace$, the functions $\neg f_n$ are computable by a sequence of polynomial-size NBPs.
Now we use an argument similar to the one in the proof of Theorem~\ref{Depth1:PE}. We apply Theorem~\ref{thm:deg_2}~(\emph{i}) and Lemma~\ref{lemma:hyper:degree2} to $\neg f_n$ and obtain a polynomial-size HGP $H_n'$ of degree~$2$ that computes $f_n$.
By Theorem~\ref{representable},
there are sequences of CQs $\q_n$ and ontologies $\T_n$ of depth~1 such that $f_n$ is a subfunction of
$f_{H^{\q_n}_{\T_n}}$ in the sense that $f_n$ is the result of replacing each $p_e$ with $1$ and each $p_v$ with the label of $v$ in $H_n'$.
By our assumption, there is a polynomial-size \FO-rewriting $\q_n'$ of $\q_n$ and $\T_n$. We eliminate the quantifiers in $\q_n'$ and apply to the result the substitution giving $f_n$ from  $\smash{f_{H^{\q_n}_{\T_n}}}$ to obtain a polynomial-size propositional Boolean formula that computes $f_n$.  Since $f_n$ is $\NLogSpace/\poly$-complete under $\NCo$ reductions, we then must have $\NLogSpace/\poly \subseteq \NCo$.
\end{proof}

As we shall see in Section~\ref{sec:tree}, \emph{tree-shaped} CQs and ontologies of depth 1 always have polynomial-size PE-rewritings.


\section{Rewritings over Ontologies of Depth 2}\label{sec:sizeof depth1}

Our next aim is to show that CQs and ontologies of depth~2 can compute the \NP-complete function checking whether a graph with $n$ vertices has a $k$-clique. We remind the reader (see, e.g., \cite{Arora&Barak09} for details) that the monotone Boolean function $\cli_{n,k}(\vec{e})$ of $n(n-1)/2$ variables $e_{jj'}$, $1 \leq j < j'\le n$,  returns 1 iff the graph with  vertices $\{1,\dots,n\}$ and  edges $\{ \{j,j'\} \mid e_{jj'}=1\}$ contains a $k$-clique.
A series of papers,  started by Razborov's~\cite{Razborov85}, gave an exponential lower bound for the size of monotone circuits computing $\cli_{n,k}$: $2^{\Omega(\sqrt{k})}$ for $k \leq \frac{1}{4} (n/ \log n)^{2/3}$~\cite{AlonB87}. For monotone formulas, an even better lower bound is known: $2^{\Omega(k)}$ for $k = 2n/3$~\cite{RazW92}.

We first construct a monotone HGP  computing $\cli_{n,k}$ and then use the intuition behind the construction to encode $\cli_{n,k}$ by means of a Boolean CQ $\q_{n,k}$ and an ontology $\T_{n,k}$ of depth 2 and polynomial size. As a consequence, any \PE- or \NDL-rewriting of $\q_{n,k}$ and $\T_{n,k}$ is of exponential size, while any \FO-rewriting is of superpolynomial size unless $\NP \subseteq \Ppoly$.

Given $n$ and $k$, let $H_{n,k}$ be a monotone HGP with vertices
\begin{align*}
& w_{jj'} \text{ labelled with } e_{jj'}, && (1 \le j < j' \le n),\\
& u_{jj'} \text{ and } u_{j'j} \text{ labelled with } 1,  && (1 \le j < j'\le n),\\
& v_i \text{  labelled with } 0 && (1 \le i \le k), 
\end{align*}
and  hyperedges
\begin{align*}
& h^{jj'} \!=\! \{w_{jj'}, u_{jj'} \} \ \ \ \text{ and } \ \ \ h^{j'j} \!=\! \{w_{jj'}, u_{j'j} \} &&  (1 \le j < j' \le n),\\
& f^{ij}  = \{v_i \} \cup \{ u_{jj'} \mid j' \ne j \} && (1 \le i \le k,  1 \le j \le n).
\end{align*}
Informally, the $w_{jj'}$ represent the edges of the complete graph with  $n$ vertices; they can be turned `on' or `off' by means of the variables $e_{jj'}$.
The vertex $u_{jj'}$ together with the hyperedge $h^{jj'}$ represent  the `half' of the edge connecting $j$ and $j'$ that is adjacent to $j$; the other `half' is represented by $u_{j'j}$ and $h^{j'j}$. The vertices $v_i$ represent a $k$-clique and the edge $f^{ij}$ corresponds to the choice of the vertex $j$ of the graph as the $i$th element of the clique.
The hypergraph of $H_{4,2}$ is shown in Fig.~\ref{fig:H42}.

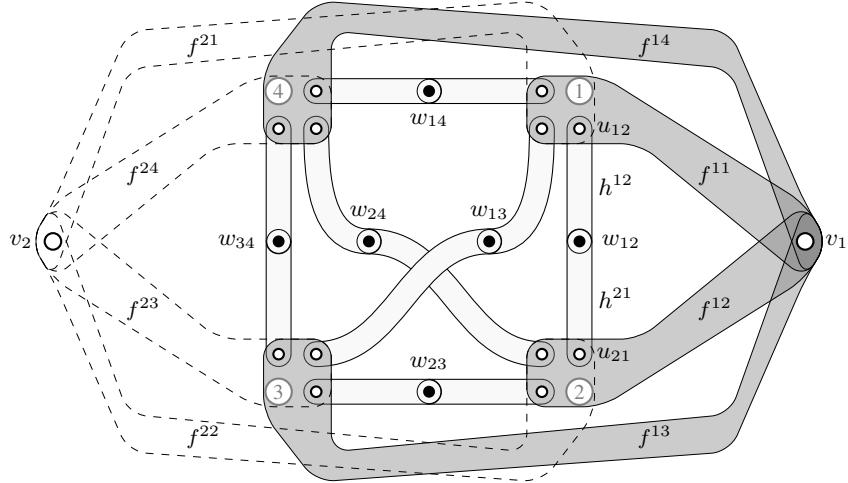
\begin{figure}[ht]
\centerline{\begin{tikzpicture}[>=latex, point/.style={circle,draw=black,fill=white,thick,minimum size=1.4mm,inner sep=0pt},
npoint/.style={circle,fill=white,draw=gray,thick,minimum size=3.5mm,inner sep=0pt},
wpoint/.style={circle,draw=black,fill=black,thick,minimum size=1.3mm,inner sep=0pt},
eline/.style={line cap=round,double=gray!5,double distance=9pt},
enode/.style={minimum size=9pt,draw=black,inner sep=0pt,ultra thin,circle}]\small
\coordinate (w12) at (4,2);
\coordinate (w23) at (2,0);
\coordinate (w34) at (0,2);
\coordinate (w14) at (2,4);
\coordinate (w24) at (1.2,2);
\coordinate (w13) at (2.8,2);
\coordinate (u12) at (4,3.5);
\coordinate (u21) at (4,0.5);
\coordinate (u23) at (3.5,0);
\coordinate (u32) at (0.5,0);
\coordinate (u34) at (0,0.5);
\coordinate (u43) at (0,3.5);
\coordinate (u41) at (0.5,4);
\coordinate (u14) at (3.5,4);
\coordinate (u24) at (3.5,0.5);
\coordinate (u42) at (0.5,3.5);
\coordinate (u13) at (3.5,3.5);
\coordinate (u31) at (0.5,0.5);
\begin{scope}
\draw[eline] (u12) -- (w12)  node[right=4pt,midway] {$h^{12}$};
\draw[eline] (u21) -- (w12)  node[right=4pt,midway] {$h^{21}$};
\node[enode] at (w12) {};
\draw[eline] (u23) -- (w23); 
\draw[eline] (u32) -- (w23); 
\node[enode] at (w23) {};
\draw[eline] (u34) -- (w34); 
\draw[eline] (u43) -- (w34); 
\node[enode] at (w34) {};
\draw[eline] (u14) -- (w14); 
\draw[eline] (u41) -- (w14); 
\node[enode] at (w14) {};
\draw[eline,out=180,in=0] (u24) to (w24); 
\draw[eline,out=-90,in=180] (u42) to (w24); 
\node[enode] at (w24) {};
\draw[eline,eline,out=-90,in=0] (u13) to (w13); 
\draw[eline,out=0,in=180] (u31) to (w13); 
\node[enode] at (w13) {};
\end{scope}
\begin{scope}[fill opacity=0.4,fill=gray,rounded corners=7,ultra thin]
\fill ($(u13)+(-0.2,-0.2)$) -- ($(u14)+(-0.2,0.2)$) -- ++(1,0) -- (7,2.5) -- (7.3,2) -- (7,1.5) -- ($(u12)+(0.8,-0.2)$) -- cycle;
\fill ($(u24)+(-0.2,0.2)$) -- ($(u23)+(-0.2,-0.2)$) -- ++(1,0) -- (7,1.5) -- (7.3,2) --(7,2.5) -- ($(u21)+(0.8,0.2)$) -- cycle;
\fill ($(u31)+(0.2,0.2)$) -- ($(u34)+(-0.2,0.2)$) -- ++(0,-1) -- (0.5,-1.2) -- (6,-0.8) -- (7,1.5) -- (7.3,2) --(7,2.5) -- (6,-0.3)  -- ($(u32)+(0.2,-0.8)$) -- cycle;
\fill ($(u42)+(0.2,-0.2)$) -- ($(u43)+(-0.2,-0.2)$) -- ++(0,1) -- (0.5,5.2) -- (6,4.8) -- (7,2.5) --(7.3,2) -- (7,1.5) -- (6,4.3)  -- ($(u41)+(0.2,0.8)$) -- cycle;
\end{scope}
\begin{scope}[rounded corners=7,ultra thin]
\draw ($(u13)+(-0.2,-0.2)$) -- ($(u14)+(-0.2,0.2)$) -- ++(1,0) -- (7,2.5) -- (7.3,2) -- (7,1.5) -- ($(u12)+(0.8,-0.2)$) -- cycle;
\draw ($(u24)+(-0.2,0.2)$) -- ($(u23)+(-0.2,-0.2)$) -- ++(1,0) -- (7,1.5) -- (7.3,2) --(7,2.5) -- ($(u21)+(0.8,0.2)$) -- cycle;
\draw ($(u31)+(0.2,0.2)$) -- ($(u34)+(-0.2,0.2)$) -- ++(0,-1) -- (0.5,-1.2) -- (6,-0.8) -- (7,1.5) -- (7.3,2) --(7,2.5) -- (6,-0.3)  -- ($(u32)+(0.2,-0.8)$) -- cycle;
\draw ($(u42)+(0.2,-0.2)$) -- ($(u43)+(-0.2,-0.2)$) -- ++(0,1) -- (0.5,5.2) -- (6,4.8) -- (7,2.5) --(7.3,2) -- (7,1.5) -- (6,4.3)  -- ($(u41)+(0.2,0.8)$) -- cycle;
\end{scope}
\begin{scope}[rounded corners=7,ultra thin,dashed]
\draw ($(u42)+(0.2,-0.2)$) -- ($(u41)+(0.2,0.2)$) -- ++(-1,0) -- (-3,2.5) -- (-3.3,2) -- (-3,1.5) -- ($(u43)+(-0.8,-0.2)$) -- cycle;
\draw ($(u31)+(0.2,0.2)$) -- ($(u32)+(0.2,-0.2)$) -- ++(-1,0) -- (-3,1.5) -- (-3.3,2) -- (-3,2.5) -- ($(u34)+(-0.8,0.2)$) -- cycle;
\draw ($(u24)+(-0.2,0.2)$) -- ($(u21)+(0.2,0.2)$) -- ++(0,-1) -- (3.5,-1.2) -- (-2,-0.8) -- (-3,1.5) -- (-3.3,2) --(-3,2.5) -- (-2,-0.3)  -- ($(u23)+(-0.2,-0.8)$) -- cycle;
\draw ($(u13)+(-0.2,-0.2)$) -- ($(u12)+(0.2,-0.2)$) -- ++(0,1) -- (3.5,5.2) -- (-2,4.8) -- (-3,2.5) -- (-3.3,2) --(-3,1.5) -- (-2,4.3)  -- ($(u14)+(-0.2,0.8)$) -- cycle;
\end{scope}
\begin{scope}\footnotesize
\node at (5,-0.6) {$f^{13}$};
\node at (5,4.6) {$f^{14}$};
\node at (5.8,2.9) {$f^{11}$};
\node at (5.8,1.1) {$f^{12}$};
\node at (-1,-0.6) {$f^{22}$};
\node at (-1,4.6) {$f^{21}$};
\node at (-1.8,2.9) {$f^{24}$};
\node at (-1.8,1.1) {$f^{23}$};
\end{scope}
\begin{scope}[label distance=1pt]
\node[point,minimum size=2.2mm,label=right:{$v_1$}] at (7,2) {};
\node[point,minimum size=2.2mm,label=left:{$v_2$}] at (-3,2) {};
\end{scope}
\node (n1) at (4,4) [npoint] {\footnotesize\textcolor{gray}{1}};
\node (n2) at (4,0) [npoint] {\footnotesize\textcolor{gray}{2}};
\node (n3) at (0,0) [npoint] {\footnotesize\textcolor{gray}{3}};
\node (n4) at (0,4) [npoint] {\footnotesize\textcolor{gray}{4}};
\begin{scope}[label distance=3pt]\small
\node at (w12) [wpoint,label=right:{$w_{12}$}] {};
\node at (w23)  [wpoint,label=above:{$w_{23}$}] {};
\node at (w34)  [wpoint,label=left:{$w_{34}$}] {};
\node at (w14)  [wpoint,label=below:{$w_{14}$}] {};
\node at (w24)  [wpoint,label=above:{$w_{24}$}] {};
\node at (w13) [wpoint,label=above:{$w_{13}$}] {};
\end{scope}
\begin{scope}[label distance=1pt]
\node at (u12)  [point,label=right:{$u_{12}$}] {};
\node at (u21)  [point,label=right:{$u_{21}$}] {};
\node at (u23)  [point] {}; 
\node at (u32)  [point] {}; 
\node at (u34)  [point] {}; 
\node at (u43)  [point] {}; 
\node at  (u41) [point] {}; 
\node at (u14)  [point] {}; 
\node at (u24)  [point] {}; 
\node at (u42)  [point] {}; 
\node at (u13) [point] {}; 
\node at (u31)  [point] {}; 
\end{scope}
\end{tikzpicture}}%
\caption{The hypergraph of $H_{4,2}$.}\label{fig:H42}
\end{figure}

\begin{theorem}\label{thm:clique}
The HGP $H_{n,k}$ computes $\cli_{n,k}$.
\end{theorem}
\begin{proof}
We show that, for each $\vec{e} \in \{0,1\}^{n(n-1)/2}$,  there is an independent set $X$ of hyperedges covering all zeros in $H_{n,k}$ under $\vec{e}$ iff $\cli_{n,k}(\vec{e}) = 1$.

\smallskip

\noindent $(\Leftarrow)$ Let $\lambda\colon \{1, \dots, k\} \to \{1, \dots, n\}$ be such that $C = \{\lambda(i) \mid 1 \leq i \leq k \}$ is a $k$-clique in the graph $G$ given by $\vec{e}$. Then
\begin{equation*}
X \ \ = \ \ \bigl\{f^{i\lambda(i)} \mid 1 \le i \le k\bigr\} \ \ \cup \ \ 
\bigl\{h^{jj'} \mid j \notin C, j' \in C\bigr\} \ \ \cup \ \
\bigl\{h^{jj'} \mid  j,j' \notin C \text{ and } j < j' \bigr\}
\end{equation*}
is independent and covers all zeros in $H_{n,k}$ under $\vec{e}$.
Indeed, $X$ is independent because, in every $h^{jj'}\in X$, the index $j$ does not belong to $C$. By definition, each $f^{i\lambda(i)}$ covers $v_i$, for $1\le i \le k$. Thus, it remains to  show that any $w_{jj'}$ with $e_{jj'}=0$ (that is, the edge $\{j,j'\}$ belongs to the complement of $G$) is covered by some hyperedge.
All edges of the complement of $G$ can be divided into two groups: those that are adjacent to $C$,
and those that are not. The $w_{jj'}$ that correspond to the edges of the former group are
covered by the $h^{jj'}$ from the middle disjunct of $X$,
where $j$ corresponds to the end of the edge $\{j,j'\}$ that is not $C$.
To cover $w_{jj'}$ of the latter group, take
$h^{jj'}$ from the last disjunct of $X$.

\smallskip

\noindent $(\Rightarrow)$ Suppose $X$ is an independent set covering all zeros labelling the vertices of $H_{n,k}$, for an input $\vec{e}$. The vertex $v_i$, $1\le i \le k$, is labelled with 0, and so  there is $\lambda(i)$ such that $f^{i\lambda(i)} \in X$.  We claim that  $C = \{ \lambda(i) \mid 1\leq i \leq k\}$ is a $k$-clique in the graph given by $\vec{e}$.
Indeed, suppose that the graph has no edge between some vertices $j,j'\in C$, that is, $e_{jj'} = 0$ for $j < j'$. Since $w_{jj'}$ is labelled with 0, it must be covered by a hyperedge in $X$, which can only be either  $h^{jj'}$ or $h^{j'j}$ (see the picture above). But $h^{jj'}$ intersects $f^{\lambda^{-1}(j)j}$ and $h^{j'j}$ intersects $f^{\lambda^{-1}(j')j'}$, which is a contradiction.
\end{proof}

\begin{figure}[ht]
\centerline{\begin{tikzpicture}[>=latex, 
query/.style={thick},
dquery/.style={densely dotted,thick},
point/.style={circle,draw=black,fill=white,thick,minimum size=1.4mm,inner sep=0pt},
zpoint/.style={circle,draw=black,fill=white,thick,minimum size=1.2mm,inner sep=0pt},
upoint/.style={circle,draw=black,fill=white,thick,minimum size=2.2mm,inner sep=0pt},
xpoint/.style={rectangle,draw=black,fill=white,thick,minimum size=1.5mm,inner sep=0pt},
npoint/.style={circle,fill=white,draw=gray,thick,minimum size=3.5mm,inner sep=0pt},
wpoint/.style={circle,draw=black,fill=white,thick,minimum size=2.2mm,inner sep=0pt}, 
eline/.style={line cap=round,draw=gray!15,double=gray!15,double distance=9pt},
enode/.style={circle,draw=black,fill=white,thick,minimum size=2.2mm,inner sep=0pt}, 
scale=1.3]\small
\coordinate (w12) at (4,2);
\coordinate (w23) at (2,0);
\coordinate (w34) at (0,2);
\coordinate (w14) at (2,4);
\coordinate (w24) at (1.2,2);
\coordinate (w13) at (2.8,2);
\coordinate (u12) at (4,3.5);
\coordinate (u21) at (4,0.5);
\coordinate (u23) at (3.5,0);
\coordinate (u32) at (0.5,0);
\coordinate (u34) at (0,0.5);
\coordinate (u43) at (0,3.5);
\coordinate (u41) at (0.5,4);
\coordinate (u14) at (3.5,4);
\coordinate (u24) at (3.5,0.5);
\coordinate (u42) at (0.5,3.5);
\coordinate (u13) at (3.5,3.5);
\coordinate (u31) at (0.5,0.5);
\begin{scope}
\draw[eline] (u12) -- (w12); 
\draw[eline] (u21) -- (w12); 
\node[enode] at (w12) {};
\draw[eline] (u23) -- (w23); 
\draw[eline] (u32) -- (w23); 
\node[enode] at (w23) {};
\draw[eline] (u34) -- (w34); 
\draw[eline] (u43) -- (w34); 
\node[enode] at (w34) {};
\draw[eline] (u14) -- (w14); 
\draw[eline] (u41) -- (w14); 
\node[enode] at (w14) {};
\draw[eline,out=180,in=0] (u24) to (w24); 
\draw[eline,out=-90,in=180] (u42) to (w24); 
\node[enode] at (w24) {};
\draw[eline,eline,out=-90,in=0] (u13) to (w13); 
\draw[eline,out=0,in=180] (u31) to (w13); 
\node[enode] at (w13) {};
\end{scope}
\begin{scope}[fill opacity=0.4,fill=gray,rounded corners=7]
\fill ($(u13)+(-0.2,-0.2)$) -- ($(u14)+(-0.2,0.2)$) -- ++(1,0) -- (7,2.5) -- (7.3,2) -- (7,1.5) -- ($(u12)+(0.8,-0.2)$) -- cycle;
\fill ($(u24)+(-0.2,0.2)$) -- ($(u23)+(-0.2,-0.2)$) -- ++(1,0) -- (7,1.5) -- (7.3,2) --(7,2.5) -- ($(u21)+(0.8,0.2)$) -- cycle;
\fill ($(u31)+(0.2,0.2)$) -- ($(u34)+(-0.2,0.2)$) -- ++(0,-1) -- (0.5,-1.2) -- (6,-0.8) -- (7,1.5) -- (7.3,2) --(7,2.5) -- (6,-0.3)  -- ($(u32)+(0.2,-0.8)$) -- cycle;
\fill ($(u42)+(0.2,-0.2)$) -- ($(u43)+(-0.2,-0.2)$) -- ++(0,1) -- (0.5,5.2) -- (6,4.8) -- (7,2.5) --(7.3,2) -- (7,1.5) -- (6,4.3)  -- ($(u41)+(0.2,0.8)$) -- cycle;
\end{scope}
\begin{scope}[rounded corners=7,ultra thin,dashed,gray!80]
\draw ($(u42)+(0.2,-0.2)$) -- ($(u41)+(0.2,0.2)$) -- ++(-1,0) -- (-3,2.5) -- (-3.3,2) -- (-3,1.5) -- ($(u43)+(-0.8,-0.2)$) -- cycle;
\draw ($(u31)+(0.2,0.2)$) -- ($(u32)+(0.2,-0.2)$) -- ++(-1,0) -- (-3,1.5) -- (-3.3,2) -- (-3,2.5) -- ($(u34)+(-0.8,0.2)$) -- cycle;
\draw ($(u24)+(-0.2,0.2)$) -- ($(u21)+(0.2,0.2)$) -- ++(0,-1) -- (3.5,-1.2) -- (-2,-0.8) -- (-3,1.5) -- (-3.3,2) --(-3,2.5) -- (-2,-0.3)  -- ($(u23)+(-0.2,-0.8)$) -- cycle;
\draw ($(u13)+(-0.2,-0.2)$) -- ($(u12)+(0.2,-0.2)$) -- ++(0,1) -- (3.5,5.2) -- (-2,4.8) -- (-3,2.5) -- (-3.3,2) --(-3,1.5) -- (-2,4.3)  -- ($(u14)+(-0.2,0.8)$) -- cycle;
\end{scope}
\begin{scope}[label distance=5pt]
\node[point,minimum size=2.2mm,label=right:{$v_1$}] (v1) at (7,2) {};
\node[point,minimum size=2.2mm,label=left:{$v_2$}] (v2) at (-3,2) {};
\end{scope}
\node (n1) at (3,3.5) [npoint] {\normalsize\textcolor{gray}{1}};
\node (n2) at (3.5,1) [npoint] {\normalsize\textcolor{gray}{2}};
\node (n3) at (0.5,1) [npoint] {\normalsize\textcolor{gray}{3}};
\node (n4) at (1,3.5) [npoint] {\normalsize\textcolor{gray}{4}};
\begin{scope}[label distance=3pt]\small
\node (nw12) at (w12) [wpoint,label=right:{$w_{12}$}] {};
\node (nw23) at (w23)  [wpoint,label=above:{$w_{23}$}] {};
\node (nw34) at (w34)  [wpoint,label=left:{$w_{34}$}] {};
\node (nw14) at (w14)  [wpoint,label=below:{$w_{14}$}] {};
\node (nw24) at (w24)  [wpoint,label=above:{$w_{24}$}] {};
\node (nw13) at (w13) [wpoint,label=above:{$w_{13}$}] {};
\end{scope}
\begin{scope}[label distance=2pt]
\node (nu12) at (u12)  [upoint,label=right:{$u_{12}$}] {};
\node (nu21) at (u21)  [upoint,label=right:{$u_{21}$}] {};
\end{scope}
\node[xpoint, label=right:{$x_{12}$}] (x12) at ($(u12)!.5!(w12)$) {};
\node[xpoint, label=right:{$x_{21}$}] (x21) at ($(u21)!.5!(w12)$) {};
\draw[->,query] (nu21) to node[midway, right] {\small $Q$} (x21);
\draw[->,query] (nw12) to node[midway, right] {\small $P_{21}$} (x21);
\draw[->,query] (nu12) to node[midway, right] {\small $Q$} (x12);
\draw[->,query] (nw12) to node[midway, right] {\small $P_{12}$} (x12);
\foreach \j/\jp in {2/3,3/4,1/4} {
\node (nu\j\jp) at (u\j\jp)  [upoint] {}; 
\node (nu\jp\j) at (u\jp\j)  [upoint] {}; 
\node[xpoint] (x\j\jp) at ($(u\j\jp)!.5!(w\j\jp)$) {};
\node[xpoint] (x\jp\j) at ($(u\jp\j)!.5!(w\j\jp)$) {};
\draw[->,query] (nu\jp\j) to (x\jp\j);
\draw[->,query] (nw\j\jp) to  (x\jp\j);
\draw[->,query] (nu\j\jp) to (x\j\jp);
\draw[->,query] (nw\j\jp) to  (x\j\jp);
}
\node (nu24) at (u24)  [upoint] {}; 
\node (nu42) at (u42)  [upoint] {}; 
\node[xpoint] (x24) at ($(u24)!.5!(w24)$) {};
\node[xpoint] (x42) at ($(u42)!.5!(w24) - (0.3,0)$) {};
\draw[->,query] (nu42) to (x42);
\draw[->,query,out=180,in=-80] (nw24) to  (x42);
\draw[->,query,out=180,in=-45] (nu24) to (x24);
\draw[->,query,out=0,in=135] (nw24) to  (x24);
\node (nu13) at (u13)  [upoint] {}; 
\node (nu31) at (u31)  [upoint] {}; 
\node[xpoint] (x13) at ($(u13)!.5!(w13) + (0.3,0)$) {};
\node[xpoint] (x31) at ($(u31)!.5!(w13)$) {};
\draw[->,query,out=0,in=-135] (nu31) to (x31);
\draw[->,query,out=180,in=45] (nw13) to  (x31);
\draw[->,query] (nu13) to (x13);
\draw[->,query,out=0,in=-100] (nw13) to  (x13);
\node[zpoint,label=above:{$z_{12}$}] (z12) at (5.5,0.9) {};
\draw[->,query,in=-160,out=-45] (nu21) to node[midway, above] {\small $U$} (z12);
\draw[->,query,rounded corners=7pt] (nu24) -- ++(0.8,-0.4) -- (z12);
\draw[->,query,rounded corners=7pt] (nu23) -- ++(1,0) -- (z12);
\draw[->,query] (v1) to node[midway,above,sloped] {\small $T_{12}$} (z12);
\node[zpoint,label=below:{$z_{11}$}] (z11) at (5.5,3.1) {};
\draw[->,query,in=160,out=45] (nu12) to node[midway, below] {\small $U$} (z11);
\draw[->,query,rounded corners=7pt] (nu13) -- ++(0.8,0.4) -- (z11);
\draw[->,query,rounded corners=7pt] (nu14) -- ++(1,0) -- (z11);
\draw[->,query] (v1) to node[midway,below,sloped] {\small $T_{11}$} (z11);
\node[zpoint,label=above left:{$z_{14}$}] (z14) at (6,4.5) {};
\draw[->,query,rounded corners=7pt] (nu41) -- ++(0.2,0.9) -- (z14);
\draw[->,query,rounded corners=7pt] (nu42) -- ++(-0.2,0.5) -- ++(0.2,1) -- (z14);
\draw[->,query,rounded corners=7pt] (nu43) -- ++(0.4,1.6) -- (z14);
\draw[->,query] (v1) to (z14);
\node[zpoint,label=below left:{$z_{13}$}] (z13) at (6,-0.5) {};
\draw[->,query,rounded corners=7pt] (nu32) -- ++(0.2,-0.9) -- (z13);
\draw[->,query,rounded corners=7pt] (nu31) -- ++(-0.2,-0.5) -- ++(0.2,-1) -- (z13);
\draw[->,query,rounded corners=7pt] (nu34) -- ++(0.4,-1.6) -- (z13);
\draw[->,query] (v1) to (z13);
\node[zpoint,label=above:{$z_{23}$}] (z23) at (-1.5,0.9) {};
\draw[->,dquery,in=-20,out=-135] (nu34) to  (z23);
\draw[->,dquery,rounded corners=7pt] (nu31) -- ++(-0.8,-0.4) -- (z23);
\draw[->,dquery,rounded corners=7pt] (nu32) -- ++(-1,0) -- (z23);
\draw[->,dquery] (v2) to (z23);
\node[zpoint,label=below:{$z_{24}$}] (z24) at (-1.5,3.1) {};
\draw[->,dquery,in=20,out=135] (nu43) to  (z24);
\draw[->,dquery,rounded corners=7pt] (nu42) -- ++(-0.8,0.4) -- (z24);
\draw[->,dquery,rounded corners=7pt] (nu41) -- ++(-1,0) -- (z24);
\draw[->,dquery] (v2) to  (z24);
\node[zpoint,label=above right:{$z_{21}$}] (z21) at (-2,4.5) {};
\draw[->,dquery,rounded corners=7pt] (nu14) -- ++(-0.2,0.9) -- (z21);
\draw[->,dquery,rounded corners=7pt] (nu13) -- ++(0.2,0.5) -- ++(-0.2,1) -- (z21);
\draw[->,dquery,rounded corners=7pt] (nu12) -- ++(-0.4,1.6) -- (z21);
\draw[->,dquery] (v2) to (z21);
\node[zpoint,label=below right:{$z_{22}$}] (z22) at (-2,-0.5) {};
\draw[->,dquery,rounded corners=7pt] (nu23) -- ++(-0.2,-0.9) -- (z22);
\draw[->,dquery,rounded corners=7pt] (nu24) -- ++(0.2,-0.5) -- ++(-0.2,-1) -- (z22);
\draw[->,dquery,rounded corners=7pt] (nu21) -- ++(-0.4,-1.6) -- (z22);
\draw[->,dquery] (v2) to (z22);
\end{tikzpicture}}%
\caption{The CQ $\q_{4,2}$ for $H_{4,2}$.}\label{fig:H42:query}
\end{figure}
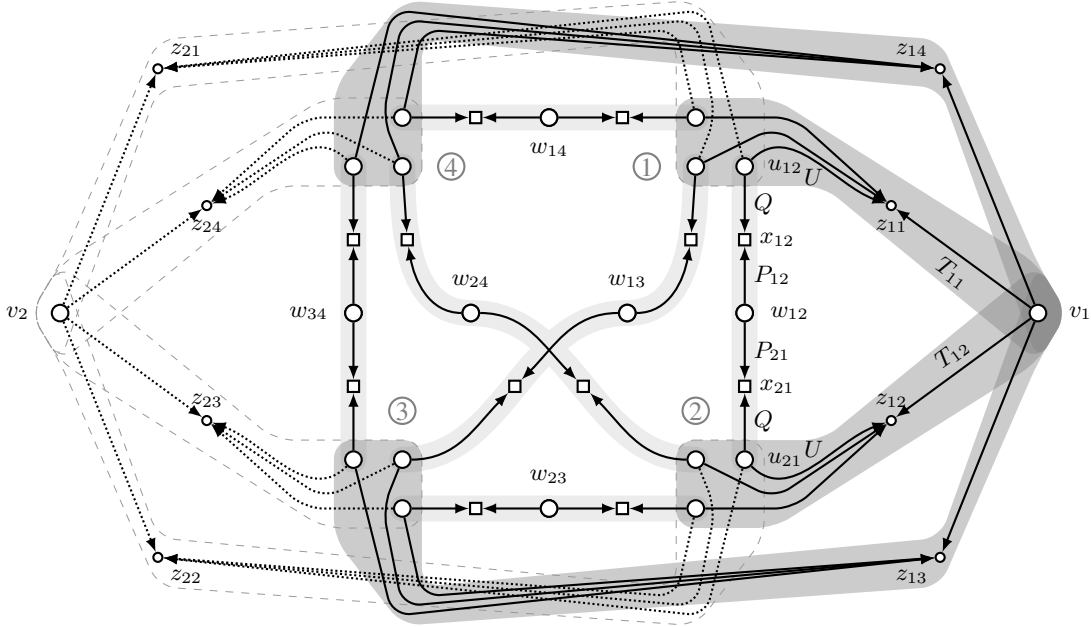

We are now in a position to define $\T_{n,k}$ of depth~2 and $\q_{n,k}$, both of polynomial size in $n$, that can compute $\cli_{n,k}$. Let $\q_{n,k}$ contain the following atoms (all variables are quantified):
\begin{align*}
& T_{ij}(v_i, z_{ij}) && (1 \le i \le k,  \ \ 1 \le j \le n),\\
& P_{jj'}(w_{jj'}, x_{jj'}), \ \ P_{j'j}(w_{jj'}, x_{j'j}) &&  (1 \le j < j' \le n),\\
& Q(u_{jj'}, x_{jj'}), \  U(u_{jj'}, z_{ij}) && (1 \le j\neq j' \le n, \ \ 1 \le i \le k).
\end{align*}
Figs.~\ref{fig:H42:query} and~\ref{fig:H42bis} show two different views of the CQ $\q_{4,2}$ for $H_{4,2}$.
Fig.~\ref{fig:qnk} illustrates the fragments of $\q_{n,k}$ centred in each variable of the form $z_{ij}$ and $x_{jj'}$ (the fragment centred in $x_{j'j}$ is similar to that of $x_{jj'}$ except the index of the $w_{jj'}$).

\begin{figure}[t]
\centerline{%
\hbox to 0pt{\includegraphics[height=60mm]{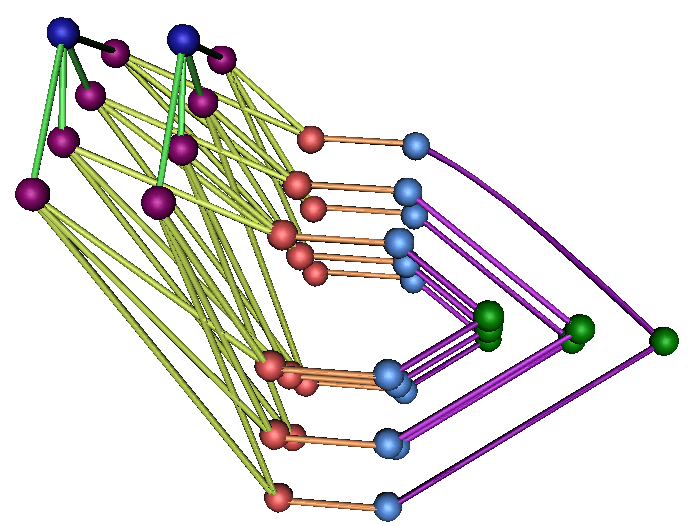}}
\hspace*{-8mm}\begin{tikzpicture}\small
\draw[white] (-1,0) rectangle (9,6);
\node at (-0.5,4) {$z_{ij}$};
\node at (-0.2,5.6) {$v_i$};
\node at (3.5,5) {$u_{jj'}$};
\node at (4.5,5) {$x_{jj'}$};
\node at (6.5,3.5) {$w_{jj'}$};
\end{tikzpicture}}%
\caption{The CQ $\q_{4,2}$ for $H_{4,2}$.}\label{fig:H42bis}
\end{figure}
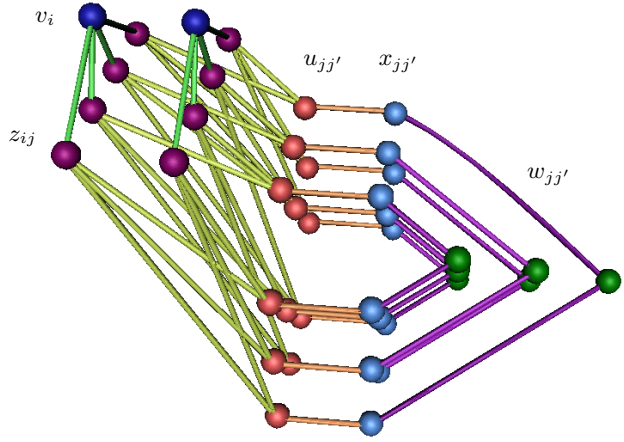

\begin{figure}[ht]
\centerline{%
\begin{tikzpicture}[>=latex, point/.style={circle,draw=black,thick,minimum size=1.2mm,inner sep=0pt,fill=white},
spoint/.style={rectangle,draw=black,thick,minimum size=1.2mm,inner sep=0pt,fill=white},
ipoint/.style={circle,draw=black,thick,minimum size=2.2mm,inner sep=0pt,fill=white},
wiggly/.style={thick,decorate,decoration={snake,amplitude=0.3mm,segment length=2mm,post length=1mm}},
query/.style={thick}]\small
\draw[ultra thin,rounded corners=8,fill=gray!40] (-0.7,-0.5) rectangle +(4.9,2.9);
\draw[ultra thin,dashed] (0.5,2) -- (-1.6,2);
\draw[ultra thin,dashed] (3,0) -- (-1.6,0);
\draw[ultra thin,dashed] (3,-2) -- (-1.6,-2);
\node (vi) at (-0.2,0) [ipoint, label=right:{\raisebox{12pt}{$v_i$}}] {};
\node (fin) at (0.5,2) [point, label=right:{$z_{ij}$}] {};
\node (ujj) at (2.2,0) [ipoint] {};
\node at (2.6,0) [ipoint] {};
\node (ujje) at (3,0) [ipoint, label=right:{$u_{jj'}$}, label=above:{\footnotesize\hspace*{2em} for $j' \ne j$}] {};
\node (xjj) at (2.2,-2) [spoint] {};
\node at (2.6,-2) [spoint] {};
\node (xjje) at (3,-2) [spoint, label=right:{$x_{jj'}$}] {};
\node (zijof) at (-0.55,-2) [point] {};
\node at (-0.3,-2) [point, label=below:{$\begin{array}{c}z_{ij''}\\[-1pt]\scriptstyle\text{for }j'' \ne j\end{array}$}] {};
\node (zijof2) at (-0.05,-2) [point] {};
\node (zijo) at (0.7,-2) [point] {};
\node at (0.95,-2) [point, label=below:{$\begin{array}{c}z_{i'j}\\[-2pt]\scriptstyle \text{for }i'\ne i\end{array}$}] {};
\node (zijoe) at (1.2,-2) [point] {};
\draw[->,thick] (vi) to node[midway, sloped, above] {\scriptsize $T_{ij}$} (fin);
\draw[->,thick] (ujje) to node[midway, sloped, above] {\scriptsize $U$} (fin);
\draw[->,thick, densely dotted] (ujj) to   (fin);
\draw[->,thick] (ujje) to node[midway, sloped, above] {\scriptsize $Q$} (xjje);
\draw[->,thick, densely dotted] (ujj) to   (xjj);
\draw[->,thick] (ujje) to  (zijoe);
\draw[->,thick] (ujje) to  (zijo);
\draw[->,thick,densely dotted] (ujj) to node[midway, sloped, above] {\scriptsize $U$}  (zijo);
\draw[->,thick,densely dotted] (ujj) to (zijoe);
\draw[->,thick,dashed] (vi) to (zijof);
\draw[->,thick,dashed] (vi) to node[midway, sloped, below, rotate=180] {\scriptsize $T_{ij''}$}  (zijof2);
\node (a) at (-1.6,-2) [point,minimum size=1.2mm,fill=black, label=below:{ $\C_{\T_{n,k}}^{A_{ij}(a)}$}] {};
\node (c) at (-1.6,0) [point,minimum size=2.2mm,thick,fill=white,label=left:{\small $c_{ij}$}] {};
\node (cp) at (-1.6,2) [point,minimum size=1.2mm,thick,fill=white,label=left:{\small $c'_{ij}$}] {};
\draw[->,wiggly] (a) to node[midway, above,sloped] {\scriptsize$T_{ij''}^-$} node[midway, below, sloped] {\scriptsize $U^-, Q^-$} (c) ;
\draw[->,wiggly] (c) to node[midway, above,sloped] {\scriptsize$T_{ij}, U$} (cp) ;
\begin{scope}[xshift=-5mm]
\draw[ultra thin,rounded corners=8,fill=gray!5] (4.9,-0.5) rectangle +(3.2,2.9);
\draw[ultra thin,dashed] (6.5,2) -- (9,2);
\draw[ultra thin,dashed] (6,0) -- (9,0);
\draw[ultra thin,dashed] (5.5,-2) -- (9,-2);
\node (wjj) at (6,0) [ipoint,  label=left:{$w_{jj'}$}] {};
\node (hjj) at (6.5, 2) [spoint, label=left:{\footnotesize $x_{jj'}$}] {};
\node (ujj) at (7, 0) [ipoint,  label=right:{\raisebox{14pt}{$u_{jj'}$}}] {};
\node (xjj) at (5.5, -2) [spoint, label=below:{$x_{j'j}$}] {};
\node (zjj) at (7,-2) [point] {};
\node at (7.25,-2) [point, label=below:{$\begin{array}{c}z_{ij}\\[-2pt]\scriptstyle\text{for all } i\end{array}$}] {};
\node (zjje) at (7.5,-2) [point] {};
\draw[->, thick] (wjj) to node[midway, sloped, above] {\scriptsize $P_{jj'}$} (hjj);
\draw[->, thick] (wjj) to node[midway, sloped, above] {\scriptsize $P_{j'j}$} (xjj);
\draw[->, thick] (ujj) to node[midway, sloped, above] {\scriptsize $Q$} (hjj);
\draw[->, thick] (ujj) to node[midway, sloped, above] {\scriptsize $U$} (zjje);
\draw[->, thick] (ujj) to   (zjj);
\node (a) at (9,-2) [point,minimum size=1.2mm,fill=black, label=below:{$\C_{\T_{n,k}}^{B_{jj'}(a)}$}] {};
\node (c) at (9,0) [point,minimum size=2.2mm,thick,fill=white,label=right:{\small $d_{jj'}$}] {};
\node (cp) at (9,2) [spoint,minimum size=1.2mm,thick,fill=white,label=right:{\small $d'_{jj'}$}] {};
\draw[->,wiggly] (a) to node[midway, below,sloped] {\scriptsize $P_{j'j}^-$} node[midway, above,sloped] {\scriptsize $U^-$} (c);
\draw[->,wiggly] (c) to node[midway, below,sloped, rotate=180] {\scriptsize $P_{jj'}, Q$} (cp);
\end{scope}
\end{tikzpicture}%
}%
\caption{Fragments of $\q_{n,k}$ and the canonical models $\C_{\T_{n,k}}^{A_{ij}(a)}$ and $\C_{\T_{n,k}}^{B_{jj'}(a)}$.}\label{fig:qnk}
\end{figure}
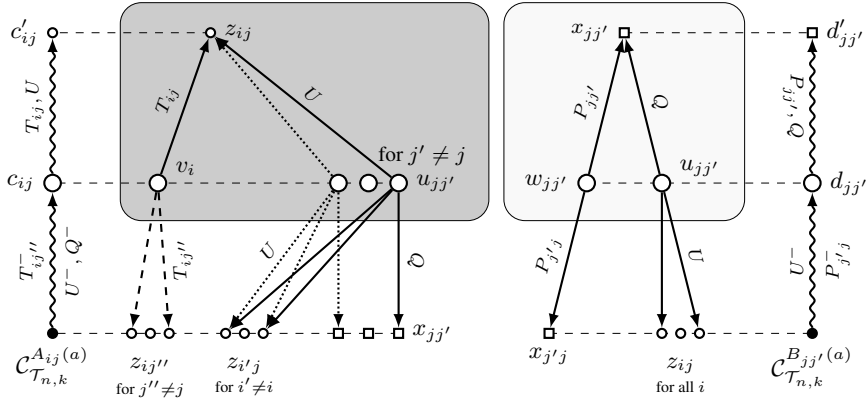

The ontology $\T_{n,k}$ mimics the arrangement of atoms in the layers depicted in Fig.~\ref{fig:qnk}
and contains the following tgds, where $1 \le i \le k$ and $1 \le j\neq j' \le n$,
\begin{align*}
A_{ij}(x) &\to  \exists y\, \bigl[\bigwedge_{j'' \ne j} T_{ij''}(y,x) \land U (y,x) \land Q(y,x) \land A'_{ij}(y)\bigr],\\
A'_{ij}(x) &\to  \exists y \, \bigl[  T_{ij}(x,y) \land  U(x,y)\bigr],\\
B_{jj'}(x) &\to \exists y \,\bigl[ P_{j'j}(y,x) \land  U(y,x) \land B'_{jj'}(y)\bigr],\\
B'_{jj'}(x) &\to  \exists y\,\bigl[ P_{jj'}(x,y) \land Q(x,y)\bigr].
\end{align*}
The canonical models $\C_{\T_{n,k}}^{A_{ij}(a)}$ and $\C_{\T_{n,k}}^{B_{jj'}(a)}$ are also illustrated in Fig.~\ref{fig:qnk} with the horizontal dashed lines showing possible ways of embedding the fragments of $\q_{n,k}$ into them.
These embeddings give rise to the following tree witnesses:
\begin{itemize}
\item[--] $\t^{ij} = (\tr^{ij}, \ti^{ij})$ generated by $A_{ij}(x)$, for $1 \le i \le k$ and  $1\le j \le n$, where
\begin{align*}
\tr^{ij}  \ \ &= \ \ \{ z_{ij'}, x_{jj'} \mid 1\le j'\le n, \ j'\neq j\} \ \ \cup \ \  \{z_{i'j} \mid 1\le i'\le k, \ i\ne i' \},\\
\ti^{ij} \ \ &= \ \ \{v_i, z_{ij} \} \cup \{u_{jj'}\mid 1\le j'\le n, \ j'\neq j\};
\end{align*}
\item[--] $\s^{jj'} = (\sr^{jj'}, \si^{jj'})$ and $\s^{j'j} = (\sr^{j'j}, \si^{j'j})$, generated by $B_{jj'}(x)$ and $B_{j'j}(x)$,
respectively,  for  $1 \le j < j' \le n$,
where
\begin{align*}
& \sr^{jj'} =  \{x_{j'j}\}\cup\{ z_{ij} \mid 1\le i\le k\},
& \si^{jj'} = \{ w_{jj'}, u_{jj'}, x_{jj'}\},\\
& \sr^{j'j} = \{x_{jj'}\}\cup\{ z_{ij'} \mid 1\le i\le k\} ,
& \si^{j'j} = \{ w_{jj'}, u_{j'j}, x_{j'j}\}.
\end{align*}
\end{itemize}
The tree witnesses $\t^{ij}$, $\s^{jj'}$ and $\s^{j'j}$ are uniquely determined by their most remote (from the root) variables, $z_{ij}$, $x_{jj'}$ and $x_{j'j}$, respectively, and  correspond to the hyperedges $f^{ij}$, $h^{jj'}$, $h^{j'j}$ of $H_{n,k}$; their internal variables of the form $v_i$, $w_{jj'}$ and $u_{jj'}$ correspond to the vertices in the respective hyperedge (see Fig.~\ref{fig:H42:query}).

Given a vector $\vec{e}$ representing a graph with $n$ vertices,
we construct  a data instance $\A_{\vec{e}}$ with a single individual $a$ by taking the following atoms:
\begin{multline*}
Q(a,a), \qquad U(a,a), \qquad  A_{ij}(a), \text{ for } 1 \le i \le k \text{ and } 1 \le j \le n,\\
P_{jj'}(a,a) \text{ and } P_{j'j}(a,a),  \text{ for } 1 \leq j  < j' \leq n \text{ with } e_{jj'} = 1.
\end{multline*}

\begin{lemma}\label{lemma:clique}
$\T_{n,k}, \A_{\vec{e}} \models \q_{n,k}$  iff $\cli_{n,k}(\vec{e}) = 1$.
\end{lemma}
\begin{proof}
$(\Rightarrow)$ Suppose $\T_{n,k}, \A_{\vec{e}} \models \q_{n,k}$. Then
there is a homomorphism $g$ from $\q_{n,k}$ to the canonical model $\C$ of $(\T_{n,k},\A_{\vec{e}})$. Since the only points of $\C$  that belong to $\exists y\,T_{ij}(x,y)$ are of the form $c_{ij}$ (see Fig.~\ref{fig:qnk}) and $\q_{n,k}$ contains atoms of the form $T_{ij}(v_i,z_{ij})$, there is $\lambda\colon \{1,\dots,k\} \to \{1,\dots,n\}$ such that $g(v_i) = c_{i\lambda(i)}$. We claim that $C = \{ \lambda(i) \mid 1 \leq i \leq k\}$ is a $k$-clique in the graph given by $\vec{e}$.

We first show that $\lambda$ is injective.
Suppose to the contrary that $\lambda(i) = \lambda(i') = j$, for $i\ne i'$. Since $\q_{n,k}$ contains $T_{ij}(v_i,z_{ij})$ and $T_{i'j}(v_{i'},z_{i'j})$, we have $g(z_{ij}) = c'_{ij}$ and $g(z_{i'j}) = c'_{i'j}$. Take $j'\ne j$. Since $U(u_{jj'},z_{ij}), U(u_{jj'},z_{i'j})\in \q_{n,k}$, we obtain $g(u_{jj'}) = c_{ij}$ and $g(u_{jj'}) = c_{i'j}$, contrary to $i \ne i'$.

Next, we show that $e_{jj'} = 1$, for all $j,j'\in C$ with $j < j'$.
Since $U(u_{jj'}, z_{ij})$ is in $\q_{n,k}$, we have $g(u_{jj'}) = c_{ij}$, and so
$g(x_{jj'}) = a$. Similarly, we also have
$g(u_{j'j}) = c_{i'j'}$ and $g(x_{j'j}) = a$. Then, since  $\q_{n,k}$ contains both $P_{jj'}(w_{jj'}, x_{jj'})$ and
$P_{j'j}(w_{jj'}, x_{j'j})$ and $\C$ contains no pair of points in both $P_{jj'}$ and $P_{j'j}$ apart from $(a,a)$, we obtain $e_{jj'} = 1$ whenever $g(x_{jj'}) = g(x_{j'j}) = a$, as shown in Fig.~\ref{fig:cnk:edge}.

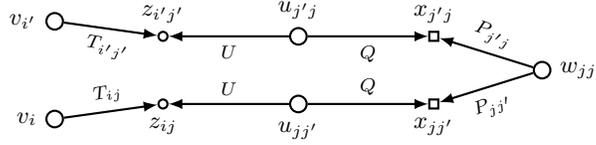
\begin{figure}[ht]
\centerline{%
\begin{tikzpicture}[>=latex, point/.style={circle,draw=black,thick,minimum size=1.2mm,inner sep=0pt,fill=white},
spoint/.style={rectangle,draw=black,thick,minimum size=1.2mm,inner sep=0pt,fill=white},
ipoint/.style={circle,draw=black,thick,minimum size=2.2mm,inner sep=0pt,fill=white},
query/.style={thick},xscale=1.2]\small
\node (w) at (4.2,0) [ipoint,label=right:{$w_{jj'}$}] {};
\node (xjpj) at (3,0.45) [spoint,label=above:{$x_{j'j}$}] {};
\node (ujpj) at (1.5,0.45) [ipoint,label=above:{$u_{j'j}$}] {};
\node (zjpj) at (0,0.45) [point,label=above:{$z_{i'j'}$}] {};
\node (vip) at (-1.2,0.65) [ipoint,label=left:{$v_{i'}$}] {};
\draw[->, thick] (w) to node[midway, sloped, above] {\scriptsize $P_{j'j}$} (xjpj);
\draw[->, thick] (ujpj) to node[midway, sloped, below] {\scriptsize $Q$} (xjpj);
\draw[->, thick] (ujpj) to node[midway, sloped, below] {\scriptsize $U$} (zjpj);
\draw[->, thick] (vip) to node[midway, sloped, below] {\scriptsize $T_{i'j'}$} (zjpj);
\node (xjjp) at (3,-0.45) [spoint,label=below:{$x_{jj'}$}] {};
\node (ujjp) at (1.5,-0.45) [ipoint,label=below:{$u_{jj'}$}] {};
\node (zjjp) at (0,-0.45) [point,label=below:{$z_{ij}$}] {};
\node (vi) at (-1.2,-0.65) [ipoint,label=left:{$v_i$}] {};
\draw[->, thick] (w) to node[midway, sloped, below] {\scriptsize $P_{jj'}$} (xjjp);
\draw[->, thick] (ujjp) to node[midway, sloped, above] {\scriptsize $Q$} (xjjp);
\draw[->, thick] (ujjp) to node[midway, sloped, above] {\scriptsize $U$} (zjjp);
\draw[->, thick] (vi) to node[midway, sloped, above] {\scriptsize $T_{ij}$} (zjjp);
\end{tikzpicture}
}%
\caption{Proof of Lemma~\ref{lemma:clique}.}\label{fig:cnk:edge}
\end{figure}

\smallskip

\noindent
$(\Leftarrow)$ Suppose $\lambda\colon \{1, \dots, k\} \to \{1, \dots, n\}$ is a $k$-clique.
We construct a homomorphism $g$ from $\q_{n,k}$ to the canonical model of $(\T_{n,k},\A_{\vec{e}})$
by taking (see~Fig.~\ref{fig:qnk}), for $1 \leq i \leq k$ and $1 \leq j < j'\  \leq n$,  
\begin{align*}
g(v_i) & = c_{i\lambda(i)},\\
g(z_{ij}) & = \begin{cases} c'_{ij}, & \text{if } j = \lambda(i), \\ a & \text{otherwise}, \end{cases} &
g(w_{jj'}) & = \begin{cases}
a,  & \text{if } j,j'\in C,\\
d_{j'j}, & \text{if } j'\notin C \text{ and } j\in C,\\
d_{jj'}, & \text{otherwise}, \end{cases}
\end{align*}
and, for $1 \leq j \ne j' \leq n$,
\begin{align*}
g(u_{jj'}) & = \begin{cases}
	c_{\lambda^{-1}(j)j}, &\text{if } j \in C,\\
	d_{jj'}, & \text{if }   j \notin C, j' \in C, \\
	d_{jj'}, & \text{if } j,j' \notin C,\  j < j',\\
	a,& \text{if } j,j' \notin C,\  j' < j,
\end{cases} &
g(x_{jj'}) & = \begin{cases}
	a,  & \text{if } j \in C,  \\
	d'_{jj'}, & \text{if } j \notin C, j'\in C, \\
	d'_{jj'}, & \text{if } j,j' \notin C, \ j < j',\\
	a, & \text{if } j,j' \notin C, \ j' < j.
\end{cases}
\end{align*}
This homomorphism mimics the cover $X$ constructed for $H_{n,k}$ in the proof of Theorem~\ref{thm:clique}. The internal variables of the tree witnesses from $X$ are sent to labelled nulls, and all other points are sent to $a$.  For example, in the definition of $g(u_{jj'})$, the first case corresponds to $u_{jj'} \in f^{\lambda^{-1}(j)j} \in X$; the second and third cases to $u_{jj'} \in h^{jj'} \in X$; and in the fourth case, $u_{jj'}$ is not covered by $X$.
It follows that $\T_{n,k}, \A_{\vec{e}} \models \q_{n,k}$.
\end{proof}

\begin{theorem}\label{Depth2:Clique}
There exists  a sequence of CQs $\q_n$ and ontologies $\T_n$ of depth $2$ any \PE- and \NDL-rewritings of which are of exponential size, while any \FO-rewriting is of superpolynomial size unless $\NP \subseteq \Ppoly$.
\end{theorem}
\begin{proof}
Given a \PE-, \FO- or \NDL-rewriting $\q'_{n,k}$ of $\q_{n,k}$ and $\T_{n,k}$, we show how to construct, respectively, a monotone Boolean formula, a Boolean formula or a monotone Boolean circuit for
the function $\cli_{n,k}$ of size $|\q'_{n,k}|$.

Suppose $\q'_{n,k}$ is a \PE-rewriting of $\q_{n,k}$ and $\T_{n,k}$. We eliminate the quantifiers in $\q'_{n,k}$ by replacing first every subformula of the form $\exists x\, \psi(x)$ in $\q'_n$ with $\psi(a)$, and then
replacing each $P_{jj'}(a,a)$ and
$P_{j'j}(a,a)$ with $e_{jj'}$, each $T_{ij}(a,a)$, $A'_{ij}(a)$ and $B'_{jj'}(a)$  with~0 and each $U(a,a)$, $Q(a,a)$, $A_{ij}(a)$ and $B_{jj'}(a)$ with~1.
One can check that the resulting propositional monotone Boolean formula computes $\cli_{n,k}$.

If  $\q'_{n,k}$ is an \FO-rewriting of $\q_{n,k}$, then we eliminate the quantifiers by replacing $\exists x\, \psi(x)$ and $\forall x\, \psi(x)$ in $\q'_{n,k}$ with $\psi(a)$, and then carry
out the replacing procedure above, obtaining a propositional Boolean formula computing $\cli_{n,k}$.

If  $(\Pi, \q'_{n,k})$ is an \NDL-rewriting of $\q_{n,k}$,
we replace all the individual variables in $\Pi$ with $a$ and then perform the replacement described above. Denote the resulting propositional \NDL-program by $\Pi'$.
The program $\Pi'$ can now be transformed into a monotone Boolean circuit computing $\cli_{n,k}$: for every (propositional) variable $p$ occurring in the head of a clause in $\Pi'$,
we introduce an $\lor$-gate whose output is $p$ and inputs are the bodies of the clauses with the head $p$; and for each such body, we introduce an $\land$-gate whose inputs are the propositional variables in the body.

Now Theorem \ref{Depth2:Clique} follows from the lower bounds for monotone Boolean circuits and formulas computing $\cli_{n,k}$ given at the beginning of this section.
\end{proof}

As the function $\cli_{n,k}$ is known to be $\NP/\poly$-complete with respect to $\NCo$-reductions, we also obtain:
\begin{theorem}\label{nppoly}
There exist polynomial-size \FO-rewritings for all CQs and ontologies of depth $2$ with polynomially-many tree witnesses  
iff all functions in $\NP/\poly$ are com\-puted by polynomial-size formulas, that is, iff $\NP/\poly \subseteq \NCo$.
\end{theorem}
\begin{proof}
$(\Leftarrow)$ Suppose $\NP/\poly \subseteq \NCo$. Consider an arbitrary CQ $\q$ and an ontology $\T$ of depth 2 with polynomially-many tree witnesses. Then the hypergraph $H^{\q}_{\T}$ is of polynomial size. The hypergraph function $f_{H^{\q}_{\T}}$ is in the class $\NP/\poly$ because the problem whether there exists an  independent set of hyperedges in a hypergraph covering all zeros is in $\NP$. 
Therefore, by our assumption, $f_{H^{\q}_{\T}}$ can be computed by a polynomial-size formula, which translates 
into a polynomial-size \FO-rewriting by Theorem~\ref{NDL}.

\smallskip

\noindent $(\Rightarrow)$ Conversely, suppose that there is a polynomial-size \FO-rewriting for all CQs and ontologies of depth~$2$.
In particular, there is a polynomial-size \FO-rewriting for the CQs and ontologies of depth 2 encoding $\cli_{n,k}$ defined above. These CQs and ontologies have polynomially-many tree witnesses. Our assumption and the construction in the proof of Theorem~\ref{Depth2:Clique} provide us with a polynomial-size Boolean formula computing $\cli_{n,k}$. Since $\cli_{n,k}$ is $\NP/\poly$-complete under $\NCo$ reductions, this gives us $\NP/\poly \subseteq \NCo$.
\end{proof}


\section{Rewritings of Tree-Shaped CQs}\label{sec:tree}

A CQ is said to be \emph{tree-shaped} if its Gaifman graph is a tree.   It is well known~\cite{DBLP:conf/vldb/Yannakakis81,DBLP:journals/tcs/ChekuriR00} that tree-shaped CQs (or, more generally, CQs of bounded treewidth) can be evaluated over plain data instances in polynomial time. In contrast, the evaluation of arbitrary CQs is \NP-complete for combined complexity and $W[1]$-complete for parameterised complexity. In this section, we consider tree-shaped CQs over ontologies.

At first sight, we do not gain much by focusing on tree-shaped CQs: answering such CQs over ontologies is \NP-complete for combined complexity~\cite{KKZ-DL11}, while their PE- and NDL-rewritings can suffer an exponential blowup~\cite{icalp12}.
However, by examining the tree-witness rewriting~\eqref{rewriting0}, we see that the $\tw_\t$ formula~\eqref{tw-formula} defines a  predicate over the data that can be computed in linear time. It follows that,  for a tree-shaped $\q$, every disjunct of~\eqref{rewriting0} can also be regarded as a  tree-shaped CQ of size $\le |\q|$. So, bearing in mind that $|\Theta^{\q}_{\T}| \le 3^{|\q|}$, we obtain the following:
\begin{theorem}
Given a tree-shaped CQ $\q(\vec{x})$, an ontology $\T$, a data instance $\A$ and a tuple $\vec{a} \subseteq \ind(\A)$, the problem of deciding whether $\T, \A \models \q(\vec{a})$ is fixed-parameter tractable, with parameter $|\q|$.
\end{theorem}

Furthermore, if each variable in a tree-shaped CQ is covered by a `small' number of tree witnesses then we can obtain  polynomial-size PE- or NDL-rewritings.

\begin{example}\label{ex:tree}\em
Consider the following ontology and CQ illustrated in Fig.~\ref{fig:divide}:
\begin{align*}
\T &= \bigl\{ \, A_i(x) \to \exists y\,\bigl(R_i(x,y) \land  R_{i+1}(y,x)\bigr) \mid 1 \leq i \leq 3 \, \bigr\},\\
\q  &=  \exists y_1 \dots y_5\bigwedge_{1 \leq i \leq 4} R_i(y_i, y_{i+1}).
\end{align*}

\begin{figure}[ht]
\centerline{%
\begin{tikzpicture}[>=latex, point/.style={circle,draw=black,thick,minimum size=1.5mm,inner sep=0pt,fill=white},
spoint/.style={rectangle,draw=black,thick,minimum size=1.5mm,inner sep=0pt,fill=white},
ipoint/.style={circle,draw=black,thick,minimum size=1.5mm,inner sep=0pt,fill=white},
wiggly/.style={thick,decorate,decoration={snake,amplitude=0.3mm,segment length=2mm,post length=1mm}},
query/.style={thick},yscale=1,xscale=1]\small
\draw[ultra thin,rounded corners=8,fill=gray!15] (0.6,-0.7) rectangle +(2.8,2.4);
\draw[ultra thin,dashed,rounded corners=8] (-0.8,-0.5) rectangle +(3.35,2);
\draw[ultra thin,dashed,rounded corners=8] (4.8,-0.6) rectangle +(-3.2,2.2);
\node at (-0.3,0.8) {\small $\q_1$};
\node at (4.3,0.8) {\small $\q_2$};
\node[ipoint,label=below:{$y_1$}] (y1) at (-0.5,0) {};
\node[ipoint,label=below:{$y_2$}] (y2) at (1,0) {};
\node[ipoint,label=above:{$y_3$}] (y3) at (2,1) {};
\node[ipoint,label=below:{$y_4$}] (y4) at (3,0) {};
\node[ipoint,label=below:{$y_5$}] (y5) at (4.5,0) {};
\draw[->,query] (y1) to node[above] {\scriptsize $R_1$} (y2);
\draw[->,query] (y2) to node[above,sloped] {\scriptsize $R_2$} (y3);
\draw[->,query] (y3) to node[pos=0.6,above,sloped] {\scriptsize $R_3$} (y4);
\draw[->,query] (y4) to node[above,sloped] {\scriptsize $R_4$} (y5);
\node[ipoint,fill=black,label=right:{$A_2$}] (c1) at (6,0) {};
\node[ipoint] (c2) at (6,1) {};
\draw[->,wiggly] (c1) to node[sloped,below] {\scriptsize $R_2$} node[sloped,above] {\scriptsize $R_3^-$} (c2);
\end{tikzpicture}%
}%
\caption{CQ and ontology from Example~\ref{ex:tree}.}\label{fig:divide}
\end{figure}
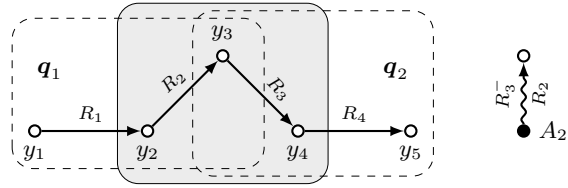

We construct a PE-rewriting $\q^\dag$ of $\q$ and $\T$ recursively by splitting $\q$ into smaller subqueries. Suppose $\T,\A \models \q$, for some $\A$. Then there is a homomorphism $h\colon \q \to \C_{\T,\A}$. Consider the `central' variable $y_3$ dividing $\q$ in half. If $h(y_3)$ is in the data part of $\C_{\T,\A}$ then $y_3$ behaves like a free variable in $\q$. Since $\q$ is tree-shaped, we can then proceed by constructing PE-rewritings, $\q_1^\dag (y_3)$ and $\q_2^\dag(y_3)$, for the subqueries
\begin{align*}
\q_1(y_3) & = \exists y_1 y_2\, (R_1(y_1,y_2) \land R_2(y_2,y_3)),\\
\q_2(y_3) & = \exists y_4 y_5\, (R_3(y_3,y_4) \land R_4(y_4,y_5)).
\end{align*}
If $h(y_3)$ is a labelled null, then $y_3$ must be  an internal point of some tree witness for $\q$ and $\T$. We have only one such tree witness, $\t=(\tr,\ti)$, generated by $A_2(x)$ with  $\tr = \{y_2,y_4\}$, $\ti = \{y_3\}$ and $\q_\t = \{ R_2(y_2,y_3), R_3(y_3,y_4)\}$ (shaded in Fig.~\ref{fig:divide}). But then $h(y_2) = h(y_4)$ and this element is in the data part of $\C_{\T,\A}$. So, we need  PE-rewritings, $\q_3^\dag (y_2)$ and $\q_4^\dag(y_4)$, of the remaining fragments of $\q$:
\begin{equation*}
\q_3(y_2) = \exists y_1 \, R_1(y_1,y_2),\qquad
\q_4(y_4) = \exists y_5\, R_4(y_4,y_5).
\end{equation*}
If the required rewritings $\q_i^\dag$, $1 \le i \le 4$, are constructed then we obtain a PE-rewriting $\q^\dag$ of $\q$ and $\T$:  
\begin{equation*}
\q^\dag \ \ = \ \ \exists y_3\, \big( \q_1^\dag (y_3) \land \q_2^\dag(y_3)\big) \ \ \lor \ \  \exists y_2 y_4 \, \big( A_2(y_2) \land (y_2 = y_4) \land\q_3^\dag (y_2) \land \q_4^\dag(y_4)\big).
\end{equation*}
We analyse $\q_1,\q_2,\q_3$ and $\q_4$ in the same way and obtain
\begin{align*}
\q_1^\dag (y_3) & = \exists y_2\,\bigl(\q_3^\dag(y_2) \land R_2(y_2,y_3)\bigr) \lor
\exists y_1\,\bigl( A_1(y_3) \land (y_1 = y_3) \bigr), \\
\q_2^\dag (y_3) & = \exists y_4\,\bigl(R_3(y_3,y_4) \land \q_4^\dag(y_4)\bigr) \lor
\exists y_5\, \bigl(A_3(y_3) \land (y_5 = y_3)  \bigr),
\end{align*}
$\q_3^\dag (y_2)$ and $\q_4^\dag(y_4)$ equal to  $\q_3(y_2)$ and $\q_4(y_4)$, respectively.
\end{example}

We now give a general definition of a PE-rewriting obtained by the strategy `divide and rewrite' and applicable to any (not necessarily tree-shaped) CQ.
Let $\q(\vec{x}) = \exists \vec{y} \, \varphi(\vec{x}, \vec{y})$ and an ontology $\T$ be given. We recursively define a PE-query $\q^\dag(\vec{x})$  as follows. Take the finest partition of $\exists \vec{y} \, \varphi(\vec{x},\vec{y})$ into a conjunction $\bigwedge_j \exists \vec{y}_j\,\varphi_j(\vec{x},\vec{y}_j)$ such that every atom containing some $y\in\vec{y}_j$ belongs to the same conjunct $\varphi_j(\vec{x},\vec{y}_j)$.
(Informally,  the Gaifman graph of $\varphi$ is cut along the answer variables $\vec{x}$.) By definition, the set of tree witnesses for $\exists\vec{y}\,\varphi(\vec{x},\vec{y})$ and $\T$ is the disjoint union of the sets of tree witnesses for the $\exists \vec{y}_j\,\varphi_j(\vec{x},\vec{y}_j)$ and $\T$.
Then we set $(\exists \vec{y} \, \varphi(\vec{x}, \vec{y}))^\dag = \bigwedge_j \psi_j$, where  $\psi_j$ is $\varphi_j(\vec{x})$ in case $\vec{y}_j$ is empty; otherwise, we choose a variable $z$ in $\vec{y}_j$ and define $\psi_j$ to be the formula
\begin{equation*}
\exists z\,\bigl(\exists\,[\vec{y}_j\setminus\{z\}] \,\varphi_j(\vec{x},\vec{y}_j)\bigr)^\dag \ \  \ \  \lor
\bigvee_{\begin{subarray}{c}\t \text{ a tree witness for } \exists \vec{y}_j\,\varphi_j(\vec{x},\vec{y}_j) \text{ and } \T\\\text{ such that } \t=(\tr,\ti) \text{ and } z\in\ti \end{subarray}} \hspace*{-5em} \exists \vec{y}_{j,\t} \,\bigl(\bigl(\exists\, [\vec{y}_j\setminus \vec{y}_{j,\t}]\, \varphi_{j,\t}(\vec{x},\vec{y}_j)\bigr)^\dag \land \tw_\t(\tr) \bigr),
\end{equation*}
where $\vec{y}_{j,\t} = \vec{y}_j \cap \tr$ contains the variables in $\vec{y}_j$  that occur among  $\tr$, the quantifiers $\exists\,[\vec{y}_j\setminus\{z\}]$  and $\exists\,[\vec{y}_j\setminus \vec{y}_{j,\t}]$ contain all variables in $\vec{y}_j$ but $z$ and $\vec{y}_{j,\t}$, respectively,
and $\varphi_{j,\t}$ consists of all the atoms of $\varphi_j$ except those in $\q_\t$. Note that the variables in $\ti$ (in particular, $z$) do not occur in the disjunct for $\t$ (and so can be removed from the respective quantifier). Intuitively, the first disjunct represents the situation where $z$ is mapped to a data individual  and treated as a free variable in the rewriting of $\varphi_j$. The other disjuncts reflect the cases where $z$ is mapped to a labelled null, and so $z$ is an internal variable of a tree witness $\t = (\tr,\ti)$ for $\exists \vec{y}_j\,\varphi_j(\vec{x},\vec{y}_j)$ and $\T$. As the variables in $\tr$ must be mapped to data individuals, this only leaves the set of atoms  $\varphi_{j,\t}$ with existentially quantified $\vec{y}_j\setminus \vec{y}_{j,\t}$ for further rewriting.  The existentially quantified variables in each of the disjuncts do not contain $z$, and so our recursion is well-founded.
The proof of the following theorem is straightforward:
\begin{theorem}\label{th:split}
For any CQ $\q(x)$ and ontology $\T$, $\q^\dag(\vec{x})$ is a PE-rewriting of $\q$ and $\T$  \textup{(}over complete data\textup{)}.
\end{theorem}

The exact form of the rewriting $\q^\dag$ depends on the choice of the variables $z$. We now consider two strategies for choosing these variables in the case of tree-shaped CQs.
Let
\begin{equation*}
d^\q_\T ~=~ 1 + \max_{z\in \vec{y}} \big|\{ \t = (\tr,\ti) \in \Theta_{\T}^{\q} \mid z\in\ti \}\big|.
\end{equation*}
We call $d^\q_\T$ the \emph{tree-witness degree  of} $\q$ and $\T$.
For example, the tree-witness degree of any CQ and ontology of depth 1 is at most 2, as observed in the proof of Theorem~\ref{depth1}. In general, however, it can only be bounded by $1 + |\Theta^\q_\T|$.

Given a tree-shaped CQ $\q(\vec{x}) = \exists \vec{y} \, \varphi(\vec{x}, \vec{y})$, we pick some variable as its root and define a partial order $\preceq$ on the variables of $\q$ by taking $z \preceq z'$ iff $z'$ occurs in the subtree of $\q$ rooted in $z$. The strategy used in~\cite{DBLP:conf/ijcai/BienvenuOSX13} chooses the smallest $z$ with respect to $\preceq$. Since the number of distinct subtrees of $\q$ is bounded by $|\q|$ and \NDL{} programs allow for structure sharing, this strategy yields  an \NDL-rewriting of size $|\T| \cdot |\q| \cdot d^\q_\T$:
\begin{corollary}[\cite{DBLP:conf/ijcai/BienvenuOSX13}]
Any tree-shaped CQ and ontology with polynomially-many tree-witnesses have a polynomial-size \NDL-rewriting.
\end{corollary}
The depth of recursion in the rewiring process with the above strategy is $|\q|$ in the worst case. Therefore, we can only obtain a \PE-rewriting of exponential size in $|\q|$. However, if we adopt the strategy of choosing $z$ that splits the graph of each $\varphi_j$ in half, then the depth of recursion does not exceed $\log |\q|$, and so the resulting \PE-rewriting is of polynomial size for $\q$ and $\T$ of bounded tree-witness degree.
This strategy is based on the following fact:
\begin{proposition}\label{PrepMiddleVertex}
Any tree $T =(V, E)$ contains a vertex $v \in V$ such that each connected component obtained by removing $v$ from $T$ has at most ${|V|}/{2}$ vertices.
\end{proposition}
As a consequence, we obtain:
\begin{theorem}\label{Depth1:Tree-shaped}
For any tree-shaped CQ $\q$ and any ontology $\T$, there is a \PE-rewriting of size $|\T| \cdot |\q|^{1 + \log d^\q_\T}$ \textup{(}over complete data\textup{)}.
\end{theorem}
\begin{proof}
Denote by $F(n)$ the maximal size of $\pq^\dag$ for a subquery $\pq$ of $\q$ with at most $n$ atoms.
We show by induction that
$F(n) \le |\T| \cdot n^{1 + \log d}$, where $d = d^\q_\T$.
By definition, for  each component $\pq_j$ of the finest partition of $\pq$, the length of its contribution to $\pq^\dag$ does not exceed
\begin{equation*}
F(n_j) + \sum\nolimits_{i = 1}^{d-1}(F(n_j - m_{ji}) + |\T| \cdot m_{ji}),
\end{equation*}
where $n_j$ is the number of atoms in $\pq_j$ and $m_{ji}$ is the number of atoms in the $i$th tree witness with $z \in \ti$, $1 \leq m_{ji} \leq n_j$. By the induction hypothesis, the length of the contribution of $\pq_j$ does not exceed
\begin{multline*}
|\T|\cdot n_j^{1 + \log d} +  |\T|\cdot  \sum_{i =1}^{d-1}\bigl((n_j-m_{ji})^{1+\log d} + m_{ji}\bigr)  \ \leq \\
|\T| \cdot \bigl(n_j^{1 + \log d} + (d - 1)\cdot n_j^{1+\log d}\bigr) = |\T| \cdot  d \cdot  n_j^{1+\log d}.
\end{multline*}
By Proposition~\ref{PrepMiddleVertex}, we can choose $z$ (at the preceding step) so that $\pq$ with $n$ atoms  is split into components $\pq_1,\dots,\pq_k$ each of which has $n_j \leq n/2$ atoms (by definition, $\sum_{j = 1}^k n_j = n$). Then we obtain
\begin{equation*}
F(n)  \ \le \  
|\T| \cdot d\cdot \sum\nolimits_{j = 1}^{k}  \hspace*{-0.2em}\bigl((n/2)^{\log d} \cdot n_j\bigr) \
\le 
|\T| \cdot n^{1+\log d},
\end{equation*}
as required.
\end{proof}

\begin{corollary}
Any tree-shaped CQ $\q$ and ontology $\T$ of depth~$1$ have a \PE-rewriting of size $|\T| \cdot |\q|^2$
\textup{(}over complete data\textup{)}.
\end{corollary}


\section{Conclusions}

We established a fundamental link between FO-rewritings of CQs over \OWLQL{} ontologies of depth 1 and 2 and---via the hypergraph functions and  programs---classical computational models for Boolean functions. This link allowed us to apply the Boolean complexity theory and  obtain both polynomial upper and exponential (or superpolynomial) lower bounds for the size of rewritings. It is to be noted that the high lower bounds were proved for CQs and ontologies with polynomially-many tree witnesses and polynomial-size chases.

A few challenging important questions remain open: (\emph{i}) Are all hypergraphs representable as subgraphs of some tree-witness hypergraphs? (\emph{ii}) Do all tree-shaped CQs have polynomial-size rewritings over ontologies of depth 2 (more generally, of bounded depth)? (\emph{iii}) What is the size of CQ rewritings over a \emph{fixed} ontology in the worst case? (The last question is related to the \emph{non-uniform} approach to the complexity of query answering in OBDA on the level of individual ontologies~\cite{DBLP:conf/kr/LutzW12}.)


\paragraph{Acknowledgment}

This work was supported by the U.K.\ EPSRC project `ExODA: Integrating Description Logics and Database Technologies for Expressive Ontology-Based Data Access' (EP/H05099X).

\bibliographystyle{abbrv}
\bibliography{DL-Lite-bib}

\end{document}